\documentclass[lettersize, journal]{IEEEtran}

\usepackage[utf8]{inputenc} % allow utf-8 input
\usepackage[T1]{fontenc}    % use 8-bit T1 fonts
\usepackage{hyperref}       % hyperlinks
\usepackage{url}            % simple URL typesetting
\usepackage{booktabs}       % professional-quality tables
\usepackage{amsfonts}       % blackboard math symbols
\usepackage{nicefrac}       % compact symbols for 1/2, etc.
\usepackage{microtype}      % microtypography
\usepackage{xcolor}         % colors
\usepackage[nocompress]{cite}
\usepackage{amsmath}
\usepackage{algorithm}
\usepackage{array}
\usepackage{multirow}
\usepackage{enumitem}
\usepackage{makecell}
\usepackage{bm}
\usepackage[pdftex]{graphicx}
\usepackage{amsthm}
\usepackage{setspace}
\usepackage{subfig}

\setlist{leftmargin=*}

\usepackage{color}
\definecolor{softorange}{RGB}{255, 120, 0}
\def\red#1{\textcolor{red}{#1}}

\long\def\comment#1{}

\def\ie{$i.e.$}
\def\eg{$e.g.$}

\newcommand{\partitle}[1]{\vspace{0.3em} \noindent \textbf{#1.}}
\hypersetup{colorlinks = true,linkcolor = blue,anchorcolor =blue,citecolor = blue,filecolor = red,urlcolor = blue, pdfauthor=author}

\newtheorem{proposition}{Proposition}

\newtheorem{definition}{Definition}

\begin{document}

\title{FIT-Print: Towards False-claim-resistant Model Ownership Verification via Targeted Fingerprint}

\author{Shuo~Shao,
        Haozhe~Zhu,
        Yiming~Li,
        Hongwei~Yao,
        Tianwei~Zhang,
        and Zhan~Qin
\thanks{Shuo Shao, Haozhe Zhu, and Zhan Qin are with the State Key Laboratory of Blockchain and Data Security, Zhejiang University, Hangzhou, 310027, China, and also with the Hangzhou High-Tech Zone (Binjiang) Institute of Blockchain and Data Security, Hangzhou, 310051, China (e-mail: \{shaoshuo\_ss, howjul, qinzhan\}@zju.edu.cn).}
\thanks{Yiming Li and Tianwei Zhang are with the College of Computing and Data Science, Nanyang Technological University, Singapore, 639798, Singapore (e-mail: liyiming.tech@gmail.com, tianwei.zhang@ntu.edu.sg).}
\thanks{Hongwei Yao is with the Department of Computer Science, City University of Hong Kong, Hong Kong, 999077, China (e-mail: yao.hongwei@cityu.edu.hk).}
\thanks{Corresponding Author: Yiming Li (liyiming.tech@gmail.com).}
\thanks{© 2026 IEEE. Personal use of this material is permitted. Permission from IEEE must be obtained for all other uses, in any current or future media, including reprinting/republishing this material for advertising or promotional purposes, creating new collective works, for resale or redistribution to servers or lists, or reuse of any copyrighted component of this work in other works.}
}

\markboth{IEEE Transactions on Information Forensics and Security}%
{Shao \MakeLowercase{\textit{et al.}}: FIT-Print}

\maketitle

\begin{abstract}
% Model fingerprinting is a widely adopted approach to safeguard the intellectual property rights of open-source models by preventing their unauthorized reuse. It is promising and convenient since it does not necessitate modifying the protected model. In this paper, we revisit existing fingerprinting methods and reveal that they are vulnerable to false claim attacks where adversaries falsely assert ownership of any third-party model. We demonstrate that this vulnerability mostly stems from their \emph{untargeted} nature, where they generally compare the outputs of given samples on different models instead of the similarities to specific references. Motivated by these findings, we propose a targeted fingerprinting paradigm (\ie, FIT-Print) to counteract false claim attacks. Specifically, FIT-Print transforms the fingerprint into a targeted signature via optimization. Building on the principles of FIT-Print, we develop bit-wise and list-wise black-box model fingerprinting methods, \ie, FIT-ModelDiff and FIT-LIME, which exploit the distance between model outputs and the feature attribution of specific samples as the fingerprint, respectively. Extensive experiments on benchmark models and datasets verify the effectiveness, conferrability, and resistance to false claim attacks of our FIT-Print.
Model fingerprinting has emerged as a crucial mechanism for safeguarding the intellectual property of open-source models, offering a non-intrusive approach that requires no modifications to the protected model. However, our analysis reveals that existing fingerprinting techniques are fundamentally vulnerable to false claim attacks, wherein adversaries can fraudulently assert ownership over independent third-party models. We demonstrate that this vulnerability stems from the untargeted nature of current methods, which evaluate model similarity based on arbitrary sample outputs rather than alignment with a specific, predefined reference. To mitigate this vulnerability, we introduce FIT-Print, a targeted fingerprinting paradigm that actively counters false claim attacks. Specifically, FIT-Print leverages optimization to transform the fingerprint into a verifiable, targeted signature. Building upon this foundation, we propose two black-box fingerprinting methods, the bit-wise FIT-ModelDiff and the list-wise FIT-LIME, which utilize output distances and feature attributions as robust model signatures, respectively. Extensive evaluations across benchmark models and datasets show that our framework perfectly neutralizes false claim attacks (100\% defense success rate) and eliminates false alarms on independent models (0.0\%), all while maintaining a 100\% ownership verification rate against diverse model reuse techniques.

% validate the effectiveness and conferrability of FIT-Print. Notably, compared to existing baselines, 
\end{abstract}

\begin{IEEEkeywords}
Model Fingerprinting, Ownership Verification, AI Copyright Protection, Trustworthy ML
\end{IEEEkeywords}

\section{Introduction}
\label{sec:intro}

% para1: 深度学习非常重要，模型具有很大的价值，深度学习模型的版权保护是一个亟待解决的问题
% 这里需要提一下Model Reused Techniques，我们需要检测的就是模型是否是reused from other parties' models。（相当于取代了model stealing的位置，model reuse会比stealing更广）

Deep learning models, particularly deep neural networks (DNNs), have achieved remarkable success across a diverse array of applications~\cite{wang2024diagnosis, zhuang2024deepreg, zhang2025smartguard}. Developing a high-performing model is inherently expensive, demanding substantial computational resources, extensive data curation, and domain expertise throughout the model development lifecycle. Despite these immense costs, many developers release their models to the open-source community (\eg, Hugging Face) to foster academic research and educational innovation. However, the development of efficient downstream model reuse techniques, such as fine-tuning~\cite{liu2018fine} and transfer learning~\cite{zhuang2020comprehensive}, introduces severe threats to the intellectual property rights (IPR) of these assets. By exploiting these methods, malicious developers can rapidly repurpose open-source models for lucrative commercial applications without authorization. Consequently, establishing reliable mechanisms to safeguard model IPR has emerged as a critical problem in trustworthy machine learning~\cite{luan2025protecting, wei2024pointncbw, gao2025toward}.

%Deep Learning (DL) models have played a crucial role in various contemporary Artificial Intelligence (AI) applications~\cite{brown2020language, li2022opboost, liu2021secure, chu2024sora}. Given the significant cost, substantial computational resources, and human expertise required for developing a well-trained model, model reuse techniques such as fine-tuning~\cite{liu2018fine} and transfer learning~\cite{zhuang2020comprehensive} have gained popularity. Model reuse can repurpose existing models and alleviate the overhead of creating new ones. However, unauthorized model reuse, such as model extraction, poses a potential threat to the intellectual property rights (IPR) of the model developers. Detecting such misbehavior is critical to protect the IPR on the valuable models.

\begin{figure*}
% \vspace{-0.5em}
    \centering
    \includegraphics[width=1\linewidth]{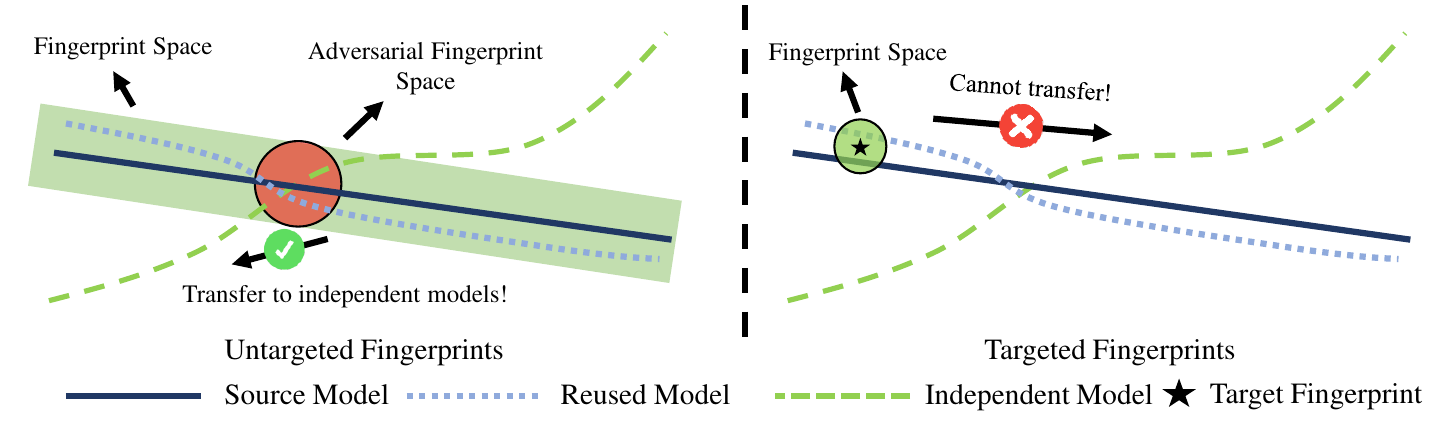}
    % \vspace{-0.8em}
    % \caption{\red{Comparison of untargeted and targeted model fingerprinting. For untargeted methods, an adversary can easily find an adversarial fingerprint space that can be transferred to independent models, leading to false claims. Targeted model fingerprints restrict the fingerprint space to be around the target fingerprint. Thus, it is hard for an adversary to find a transferable fingerprint.}}
    \caption{The comparison of untargeted and targeted fingerprinting paradigms. Untargeted methods generally compare the output of given samples. Thus, using some transferable samples can lead to false claims. Targeted fingerprinting calculates the similarity to a specific signature, which restricts the fingerprint space around the target and, therefore, mitigates false claim attacks.}
    \label{fig:untarget}
    % \vspace{-0.5em}
\end{figure*}

Currently, ownership verification is a widely adopted post-hoc approach to safeguard the IPR of model developers. This method aims to determine whether a suspicious third-party model is reused from the protected model~\cite{zhang2020passport,sun2023deep,li2025move}. Existing techniques for ownership verification generally fall into two main categories: model watermarking and model fingerprinting. Model watermarking~\cite{adi2018turning, shao2024fedtracker,li2025move} involves embedding an owner-specific signature (\ie, watermark) into the model. The model developer can then extract this watermark from the model to verify their ownership. In contrast, model fingerprinting~\cite{cao2021ipguard, li2021ModelDiff, jia2022ZestLIME, yang2022metafinger} aims to identify the intrinsic features (\ie, fingerprints) of the model rather than modifying it. A fingerprint can be represented by the outputs of specific testing samples through a particular mapping function. By comparing these fingerprints, one can check whether the suspicious model is a reused version of the source model. Arguably, model fingerprinting is more practical than model watermarking because it does not require any changes to the model's parameters, structure, or training process, and therefore has no negative impact on the model itself.

% para3: 模型指纹相关工作的介绍，分为两类，基于对抗样本和基于测试
% Existing model fingerprint methods can be classified into two categories: adversarial-example-based methods~\cite{cao2021ipguard, lukas2021DeepNeuralNetwork, wang2021characteristic, peng2022FingerprintingDeepNeural} and testing-based methods~\cite{chen2022copy, chen2022teacher, guan2022AreYouStealing, jia2022ZestLIME}. Adversarial-example-based methods assume that reused models share similar decision boundaries. As such, utilizing adversarial examples~\cite{goodfellow2014explaining, ren2020adversarial} can mark the models' decision boundaries since the adversarial examples tend to be near the decision boundaries. On the contrary, testing-based methods compare the outputs of the models on a certain metric function $M(\cdot)$. If the distance is less than a threshold, one model can be considered as being reused from the other.

% para4: 阐述模型指纹存在False Claim Attack的问题；这里需要注意的一点是模型指纹的用途不只有ownership verification，所以我们这里只claim它们在OV方面存在问题。

In this paper, we reveal that existing model fingerprinting methods, whether they are bit-wise~\cite{li2021ModelDiff} or list-wise~\cite{jia2022ZestLIME} (\ie, extracting the fingerprint bit by bit or as an entire list), are vulnerable to false claim attacks. In general, false claim attacks~\cite{liu2023false} allow adversaries to falsely claim ownership of independent third-party models by creating a counterfeit ownership certificate (\ie, a watermark or fingerprint). Specifically, these attacks can be viewed as finding transferable ownership certificates across different models, because registering a certificate with a timestamp prevents any false claims made at a later date. We show that an adversary can conduct false claim attacks by constructing transferable, "easy" samples that are correctly classified with high confidence (detailed in Section~\ref{sec:falseclaim}). Since existing fingerprinting methods typically compare the outputs of testing samples, these carefully crafted, easy samples can produce similar high-confidence outputs across various models. Consequently, this leads to independent models being misjudged as reused ones.
%Since model fingerprinting methods compare the outputs of the testing samples, we find that the adversary can construct some transferably `easy' samples that can be correctly classified with high confidence. 
%These easy samples can have similar outputs on various models, thus leading to misjudgments. 

We argue that this vulnerability primarily stems from the \emph{untargeted} nature of existing model fingerprinting methods. Specifically, they generally compare the outputs of arbitrary samples across different models, rather than measuring their similarity to specific references. This untargeted nature enlarges the space of viable fingerprints, making it easy for adversaries to find transferable samples that produce similar outputs on independent models, as illustrated in Fig.~\ref{fig:untarget}.

Motivated by these findings, we introduce a new fingerprinting paradigm, dubbed \underline{F}alse-cla\underline{I}m-resistant \underline{T}argeted model finger\underline{Print}ing (FIT-Print), where the fingerprint comparison is \emph{targeted} rather than untargeted. The key insight of FIT-Print is to shift model ownership verification from passive feature extraction to active signature alignment. Specifically, we optimize the perturbations on the testing samples so that the output of the mapping function closely aligns with a specific signature (\ie, the target fingerprint). This restricts the potential fingerprint space, significantly reducing the probability of a successful false claim attack. Based on FIT-Print, we develop two targeted model fingerprinting methods: FIT-ModelDiff and FIT-LIME, representing bit-wise and list-wise approaches, respectively. FIT-ModelDiff exploits the distances between model outputs, while FIT-LIME leverages the feature attributions of testing samples as the fingerprint.
%\red{We then propose two implementations under the FIT-Print framework, namely FIT-ModelDiff and FIT-LIME. They utilize the distances between outputs and the feature attribution of testing samples as the fingerprint}.
%Our main contributions are as follows.
%The main insight of FIT-Print is to make the model fingerprints targeted. Specifically, 

%Based on the above findings, designing a targeted model fingerprinting method is arguably the only way to defeat false claim attacks. In this paper, we propose FIT-Print, a \underline{F}alse-cla\underline{I}m-resistant \underline{T}argeted model finger\underline{Print}ing framework. 
%Feature attribution is a type of eXplainable Artificial Intelligence (XAI) method that outputs a real-valued importance score of each feature in the input sample. Specifically, by optimizing the perturbation, we turn the feature attribution of the perturbed trigger sample into an owner-specific signature. Once the model is reused by other parties, the model developer can trigger the specific explanation by adopting the feature attribution algorithm.

% para7: Contributions

Our main contributions are four-fold: 
\begin{itemize}
    \item We revisit existing model fingerprinting methods and reveal that they are generally vulnerable to false claim attacks due to their untargeted nature.
    \item We introduce a new fingerprinting paradigm (\ie, FIT-Print), which conducts ownership verification in a targeted manner using a specific reference.
    \item Based on FIT-Print, we design two black-box targeted fingerprinting methods: FIT-ModelDiff and FIT-LIME.
    \item We conduct extensive experiments on benchmark datasets and models to demonstrate their effectiveness, conferrability, and resistance to both false claims and adaptive attacks.
\end{itemize}

\section{Revisiting Existing Model Fingerprinting}
\label{sec:pre}

% 我们正式全面定义了XXX，

In this section, we first formally and comprehensively define the threat model of model fingerprinting. We then categorize existing methods into two types and describe their formulations. Subsequently, building on these definitions, we design a simple yet effective false claim attack, revealing the underlying vulnerability of current fingerprinting methods to such attacks.

\subsection{Threat Model of Model Fingerprinting}
\label{sec:threat}

In this paper, we consider three parties in our threat model, including \emph{model developer}, \emph{model reuser}, and \emph{verifier}. Arguably, including a verifier is necessary and may improve the trustworthiness of model fingerprinting, although there are currently still no mature legal provisions about this.

\partitle{Process of Model Fingerprinting} Model fingerprinting can be divided into three steps, including fingerprint generation, fingerprint registration, and ownership verification.

\begin{enumerate}
    \item \textbf{Fingerprint Generation}: The model developer trains their source model $M_o$ and generates the fingerprint of $M_o$. Broadly, a model fingerprint refers to any unique, intrinsic characteristic that allows a model to be distinguished from other independently trained models.
    \item \textbf{Fingerprint Registration}: After generating the fingerprint, the model developer registers the fingerprint and the model with a timestamp to a trustworthy third-party verifier.
    \item \textbf{Ownership Verification}: For a suspicious model $M_s$ that could be a reused version of $M_o$, the verifier will first check the timestamps of these two models. If the registration timestamp of $M_s$ is later than $M_o$, the verifier will further check whether the fingerprint of $M_o$ is similar to the fingerprint $M_s$. If so, the suspicious model can be regarded as a reused version of $M_o$.
\end{enumerate}

%One is the model developer who trains the source model and the other is the model reuser who tries to reuse the source model without authorization. 
%The model reuser's model is called the suspicious model.

\begin{figure*}
    \centering
    % \vspace{-2em}
    \includegraphics[width=1\linewidth]{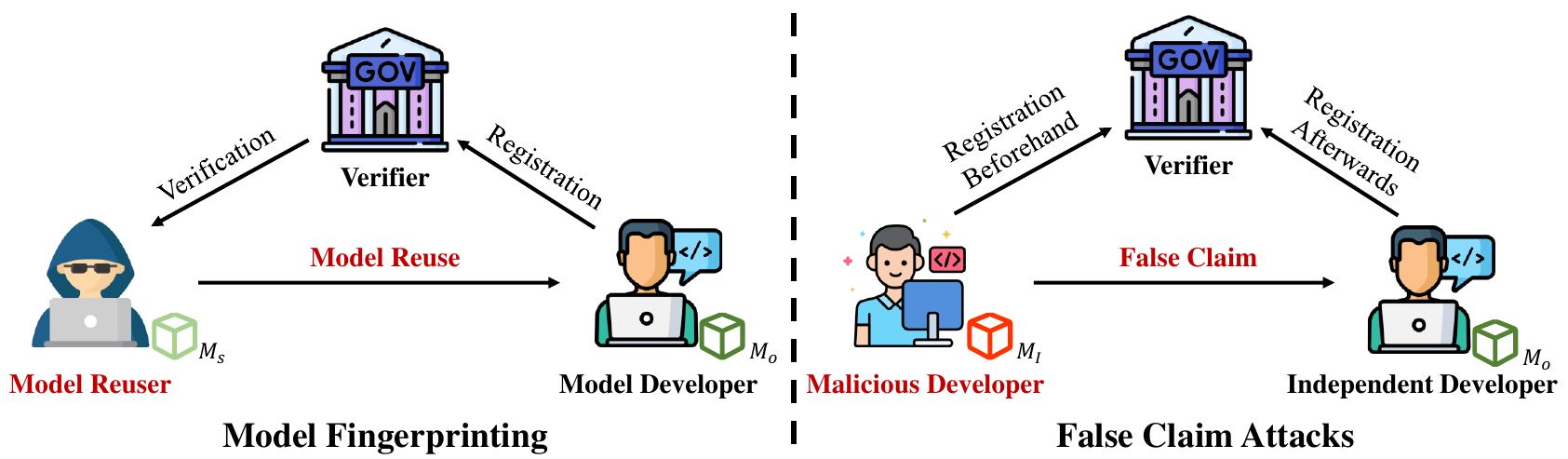}
    \caption{The threat models and detailed processes of model fingerprinting and false claim attacks. In model fingerprinting, the model developer generates and registers the model and the fingerprint with a third-party verifier. Once the model is reused by a model reuser, the verifier can determine the model ownership by comparing the fingerprints. Instead, in false claim attacks, the malicious developer attempts to register a transferable fingerprint to falsely claim other independent developers' models.}
    \label{fig:threatmodel}
    % \vspace{-1em}
\end{figure*}

\partitle{Assumptions of the Model Developer and Verifier} The model developer is the owner of the source model and can register their model and fingerprint to the trustworthy verifier with a timestamp. The verifier is responsible for fingerprint registration and verification. In case the model is reused by a model reuser, the model developer can ask the verifier for ownership verification. In particular, if two parties can simultaneously provide fingerprints and verify the ownership of a model, the fingerprint with a later timestamp will be deemed invalid. The model developer and the verifier are assumed to have \textbf{(1)} white-box access to their source model and \textbf{(2)} black-box access to the suspicious model.
%and aims to identify whether a third-party suspicious model is actually reused from its source model. The model developer is assumed to have white-box access to its source model and can register their model and its fingerprint to a trustworthy verifier with a timestamp.
%and \textbf{(2)} black-box access to the suspicious model. %which is deployed by the model reuser.

\partitle{Assumptions of the Model Reuser} Model reusers aim to avoid having their authorized reuse detected by the verifier. To achieve this, they can first modify the victim model via various techniques, such as fine-tuning, pruning, transfer learning, and model extraction, before deployment.

%\textbf{Assumptions of the Verifier.} The verifier is a trustworthy third party who is responsible for fingerprint registration and verification. In case the model is reused by a model reuser, the model developer can ask the verifier for ownership verification using the registered fingerprint. The verifier is assumed to have black-box access to the suspicious model.

\subsection{The Formulation of Existing Fingerprinting}
\label{sec:formulation}

In this section, we outline the formulations of existing fingerprinting methods to aid in the analysis and design process in the subsequent sections of this paper and follow-up research. In this paper, we focus on black-box methods since they are more practical in the real world.  In general, existing black-box model fingerprinting methods can be categorized into two types: adversarial example-based (AE-based) fingerprinting methods and testing-based fingerprinting methods. We also include a broader discussion about other fingerprinting methods in Section~\ref{sec:related}.

% \vspace{0.3em}
\partitle{AE-based Fingerprinting Methods} AE-based fingerprinting methods~\cite{cao2021ipguard, lukas2021DeepNeuralNetwork, pautov2024probabilistically} assume that the independent model has a unique decision boundary. Based on this assumption, they exploit adversarial examples (AE)~\cite{ren2020adversarial} to characterize the properties of the decision boundary of a model. 
AEs are visually similar to the benign samples but misclassified by the model~\cite{ren2020adversarial, wang2021feature}. Due to such a property, AEs tend to be near the decision boundary. Subsequently, 
AE-based fingerprinting methods validate whether the AEs are misclassified by the source model and the suspicious model. If so, the suspicious model can be treated as a reused version of the source model. The definition of the ownership verification in AE-based methods can be formulated as Definition~\ref{prop:ae-based}.

\begin{table*}[t]
% \vspace{-0.8em}
    \tabcolsep=2.5mm
    \renewcommand{\arraystretch}{1.1}
    \centering
    \caption{False claim attack against three testing-based methods. It is observed that the distances between independent models and the source model (Ind. Model Dist.) after the attack are approximately equal to or less than the average distance between the reused models and the source model (Reused Model Dist.), demonstrating the vulnerability of existing methods against false claim attacks.}
    \label{tab:falseclaim}
    \scalebox{0.85}{
    \begin{tabular}{cc|cc|cc|cc}
    \hline
    \hline
    % \toprule
    & Method$\rightarrow$  & \multicolumn{2}{c|}{ModelDiff} & \multicolumn{2}{c|}{Zest} & \multicolumn{2}{c}{SAC} \\
    & Dataset$\rightarrow$   & SDogs120   & Flowers102   & SDogs120   & Flowers102   & SDogs120   & Flowers102     \\ 
    \hline
    \multirow{2}{*}{Ind. Model Dist.} 
    & Before Attack & 0.131        & 0.114  & 0.177      & 0.161       & 0.080     & 0.094                 \\ 
    & After Attack  & 0.093        & 0.083          & 0.114      & 0.098       & 0.078     & 0.092       \\ 
    \hline
    Reused Model Dist.  & Average    & 0.108        & 0.092        & 0.095      & 0.072       & 0.079     & 0.081       \\
    \hline
    \hline
    % \bottomrule
    \end{tabular}
    }
    % \vspace{-1em}
\end{table*}

% \vspace{0.3em}
\begin{definition}[Ownership Verification of AE-based Fingerprinting Methods]
\label{prop:ae-based}
    Let $M_o$ be the source model and $M_s$ be the suspicious model, and $g(\bm{x})$ is the function that always outputs the ground-truth label of any input data $\bm{x}$. If for any $\bm{x} \in \mathcal{X}_T$ ($\mathcal{X}_T$ denotes the set of testing samples), we have
    \begin{equation}
        M_o(\bm{x})=M_s(\bm{x})\neq g(\bm{x}).
    \end{equation}
    
    The suspicious model $M_s$ can be asserted as a reused version of the source model $M_o$.
\end{definition}

%However, AE-based fingerprinting methods have several fatal drawbacks. First, AE-based methods can only apply to the models with the same task. If the model is transferred to other tasks, AE-based methods cannot detect this type of model reuse due to the difference in the output classes. Second, existing works~\cite{liu2023false} proved that AE-based methods are vulnerable to false claim attacks. We will discuss the false claim attacks in Section~\ref{sec:falseclaim}.

% \vspace{0.3em}
\partitle{Testing-based Fingerprinting Methods} Testing-based fingerprinting methods~\cite{li2021ModelDiff, jia2022ZestLIME, guan2022AreYouStealing} aim to compare the suspicious model with the source model on a specific mapping function $f(\cdot)$. If the outputs are similar, the suspicious model can be regarded as being reused from the source model. As such, the core of testing-based fingerprinting methods is how to design the mapping function $f(\cdot)$. 
The definition of ownership verification in testing-based fingerprinting methods can be formulated as Definition~\ref{prop:testing-based}.

% \vspace{0.3em}
\begin{definition}[Ownership Verification of Testing-based Fingerprinting Methods]
\label{prop:testing-based}
    Let $M_o$ be the source model and $M_s$ be the suspicious model. If for a specific mapping function $f(\cdot)$ and any $\bm{x} \in \mathcal{X}_T$ ($\mathcal{X}_T$ is the set of testing samples), we have
    \begin{equation}
        \frac{1}{|\mathcal{X}_T|}\sum_{\bm{x}\in\mathcal{X}_T}{\tt dist}(f[M_o(\bm{x})], f[M_s(\bm{x})])\leq \tau,
    \end{equation}
    where $\tau$ is a small positive threshold and ${\tt dist}(\cdot, \cdot)$ is a distance function. Then, the suspicious model $M_s$ can be asserted as reused from the source model $M_o$.
\end{definition}

% Model fingerprinting methods do not distinguish which model is the source model. Therefore, the auxiliary information is necessitated to determine. A simple yet effective way is to register the fingerprint of the model immediately after finishing the training~\cite{liu2023false, waheed2024grove}. The model with an earlier registration timestamp will be regarded as the source model.

\subsection{False Claim Attack against Model Fingerprinting}
\label{sec:falseclaim}

Existing model fingerprinting methods primarily assume that the model reuser is the adversary while paying little attention to the false claim attack~\cite{liu2023false}\footnote{The concept and definition of false claim attacks were initially introduced in~\cite{liu2023false, shao2025databench} and primarily targeted at attacking model watermarking methods.} where \textbf{the model developer is the adversary}. The formal definition of the false claim attack is as follows.

% \vspace{0.3em}
\begin{definition}
    A false claim attack refers to a malicious attempt by a malicious model developer to falsely assert the ownership of an independent model $M_I$ by registering some fraudulent testing samples $\bar{\bm{x}}$ that can pass the ownership verification of Definition~\ref{prop:ae-based} or Definition~\ref{prop:testing-based}.
\end{definition}

As depicted in Fig.~\ref{fig:threatmodel}, there are also three different parties involved in the threat model of false claim attacks, including the malicious developer, the verifier, and an independent developer. %The formal definition of false claim attacks can be found in Section~\ref{sec:falseclaim}.

% \partitle{Assumption of Malicious Developer}. 
In false claim attacks, the malicious developer is the adversary who aims to craft and register a \emph{transferable} fingerprint to falsely claim the ownership of the independent developer's model $M_I$. The malicious developer is assumed to have adequate computational resources and datasets to train a high-performance model and carefully craft transferable model fingerprints. The primary goal of the malicious developer is that the registered model fingerprints can be verified in as many other models as possible. By generating the transferable fingerprint, the malicious developer can (falsely) claim the ownership of any third-party models (that are registered later than that of the malicious developer).

\partitle{Process of False Claim Attacks} The process of false claim attacks can be divided into 3 steps, including fingerprint generation, fingerprint registration, and false ownership verification.

\begin{enumerate}%[leftmargin=*]
    \item \textbf{Fingerprint Generation}: In this step, the malicious model developer trains their source model $M_o$ and attempts to generate a \emph{transferable} fingerprint of $M_o$.
    \item \textbf{Fingerprint Registration}: After generating the fingerprint, the malicious model developer registers the \emph{transferable} fingerprint and the model with a timestamp to a trustworthy third-party verifier.
    \item \textbf{False Ownership Verification}: The malicious model developer could try to use the transferable fingerprint to falsely claim the ownership of another independently trained model $M_I$. Since the fingerprint is registered beforehand, the ownership verification won't be rejected due to the timestamp. Subsequently, the benign developer may be accused of infringement.
\end{enumerate}

%The detailed process of false claim attacks is in Appendix~\ref{sec:detailthreatfc}. 
%Some terms (\eg, ambiguity attack and false positive rate) may have a similar definition to the false claim attack. We clarify their differences in Appendix~\ref{sec:further}. 
Since registering the fingerprint with a timestamp can prevent any false claim after registration, its success hinges on generating a transferable fingerprint. For AE-based methods, \cite{liu2023false} has successfully implemented the false claim attacks by constructing transferable AEs that force diverse models to consistently misclassify them into specific wrong classes. As such, we hereby mainly focus on designing false claim attacks against cutting-edge testing-based methods. 
While our attack shares the fundamental insight of exploiting transferability, it operates in an opposite manner for testing-based fingerprinting. Instead of inducing misclassification, our primary insight is to craft transferable ``inverse-AEs'' $\bar{\bm{x}}$ as the malicious fingerprinting samples. $\bar{\bm{x}}$ are abnormally easy samples that force independent models to consistently classify them correctly with exceptionally high confidence, leading to:
%Our primary insight is to craft inverse-AEs $\bar{\bm{x}}$ as the malicious fingerprinting samples, which can be ``easily'' classified, leading to

\begin{equation}
\begin{aligned}
    &M_o(\bar{\bm{x}})\approx M_I(\bar{\bm{x}}) \\
    \Rightarrow\ & {\tt dist}(f[M_o(\bar{\bm{x}})], f[M_I(\bar{\bm{x}})]) \approx 0 \leq \tau.
\end{aligned}
\end{equation}
%Intuitively, while the metric function $f(\cdot)$ varies between different methods
%The foundation of testing-based fingerprinting methods is outlined in Proposition 2. It\textquotesingle s intuitive that while the metric function \(f(\cdot)\) varies between different methods, the parameters remain consistent. Our preliminary approach includes the construction of samples \(x_T\) that are "easily" classified, designed to achieve a confidence level near unity across various models, as illustrated by the equation:
% \begin{equation}
% M_o(x_T)\approx M_s(x_T)\approx g(x_T).
% \end{equation}
% This leads to the condition that successfully pass MOV procedure:
% \begin{equation}
% |f(M_o(x_T))-f(M_s(x_T))|\approx0\leq\tau.
% \end{equation}

To execute this strategy, motivated by the fast gradient sign method (FGSM) \cite{goodfellow2014explaining} for AE generation, we propose to leverage Eq.~(\ref{eq:fgsm}) to generate malicious fingerprinting samples, as follows:
% methods utilized in the generation of adversarial samples, which typically aim to create samples that pose classification challenges. Conversely, our goal is the opposite. Since the primary objective of this section is to prove that testing-based methods are susceptible to false claim attacks, we just adopt the simplest one. Building upon

\begin{equation}
\label{eq:fgsm}
    \bar{\bm{x}}=\bm{x}-\gamma\cdot{\tt sign}(\nabla J(M_o,\bm{x},\bm{y})),
\end{equation}
where ${\tt sign}(\cdot)$ denotes the sign function, $J(\cdot)$ represents the loss function associated with the original task of $M_o$, and $\gamma$ signifies the magnitude of the perturbation. More powerful transferable adversarial attacks can be exploited here, but we aim to show that using simple FGSM can also falsely claim to have ownership of some independent models.

\partitle{Results} We exploit 3 representative testing-based methods, \ie, ModelDiff~\cite{li2021ModelDiff}, Zest~\cite{jia2022ZestLIME}, and SAC~\cite{guan2022AreYouStealing}, to validate the effectiveness of our false claim attacks. 
%The complete results can be found in Appendix~\ref{sec:full-falseclaim}. 
As shown in Table~\ref{tab:falseclaim}, SAC is poor at identifying models of the same tasks, even without attacks. Moreover, after attacks, the distances between the source model $M_o$ and the independent model $M_I$ of all three methods are approximately equal to or less than the average distances between reused models and the source model. It indicates that $M_I$ will be asserted as reused from $M_o$, which is a false alarm. The results demonstrate that existing model fingerprinting methods are vulnerable to false claim attacks.

% Here, \(\text{sign}(\cdot)\) denotes the sign function, \(J(\cdot)\) represents the loss function associated with \(M_o\)\textquotesingle s original task, and \(\varepsilon\) signifies the magnitude of the perturbation.

% Part 3: 实验结果
% \textbf{Experiments.} We employed the method described above to conduct
% false claim attacks against three testing-based methodologies: Zest\cite{jia2022ZestLIME}, SAC\cite{guan2022AreYouStealing}, and ModelDiff\cite{li2021ModelDiff}. The \emph{Average Positive Distance} was determined by measuring several groups of positive models for the corresponding tasks. \emph{Negative Distance} gauges the distance between pairs of negative models both before and after the attack. As illustrated in Table 1, when employing easy samples, the negative distance for both Zest and ModelDiff was significantly reduced, approaching or even falling below the average positive distance. In the case of SAC, although there wasn\textquotesingle t a marked reduction in distance, many pairs of negative models were already close to or below the average positive distance before deploying the easy sample attack, so we were still able to perform a false claim attack when using the SAC method.

\begin{figure*}[t]
    \centering
    % \vspace{-0.8em}
\includegraphics[width=1\linewidth]{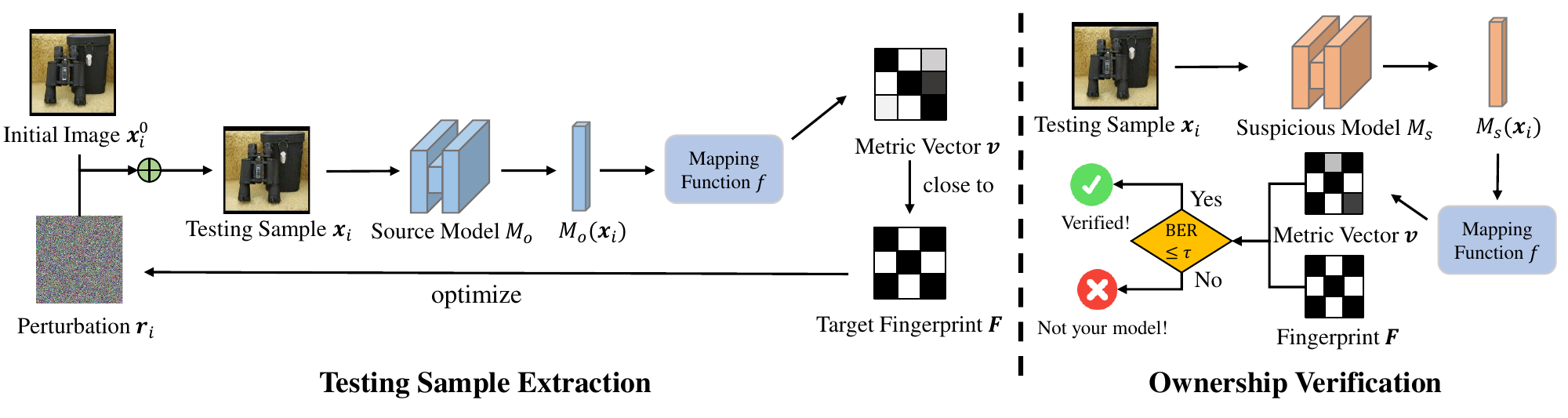}
% \vspace{-1em}
    \caption{The pipeline of FIT-Print. In testing sample extraction, FIT-Print optimizes the perturbations to turn the fingerprint vector close to the target fingerprint. In the ownership verification stage, FIT-Print extracts the fingerprint from the suspicious model and compares it with the original fingerprint.}
    \label{fig:pipeline}
    % \vspace{-1em}
\end{figure*}

\section{Methodology}
\label{sec:method}

\subsection{Design Objectives}
\label{sec:objective}

% Effectiveness, Conferrability, Non-transferability, or Resistance to False Claim Attack
%In this section, we present the design objectives of the model fingerprinting methods for ownership verification. 
The objectives of model fingerprinting methods can be summarized in three-fold, including effectiveness, conferrability, and resistance to false claim attacks.

\begin{itemize}
    \item \textbf{Effectiveness:} Effectiveness indicates that the model developer can successfully verify the ownership of the source model using the fingerprinting method.
    \item \textbf{Conferrability:} Conferrability requires the model fingerprint to be conferrable to models reused from the source model. In other words, the fingerprints of the source and reused models should be highly similar.
    \item \textbf{Resistance to False Claim Attacks:} This requires that independently trained models have distinctly different fingerprints. Additionally, it prevents malicious developers from constructing a transferable fingerprint that can be extracted from independent models.
\end{itemize}

\subsection{The Insight of our FIT-Print}
% \subsection{The Main Pipeline of FIT-Print}
\label{sec:framework}
% 这里叫main pipeline不是特别合适，因为我这里实际上没有在介绍pipeline，介绍的是FIT-Print的一个核心命题：我们检验的是模型提取出的指纹是否接近一个特定的signature，所以叫insight或者motivation会更合适一点。

% Part 1: 总结一下之前讨论的结果，阐述我们的主要思想。
As discussed in Section~\ref{sec:falseclaim}, existing model fingerprinting methods are vulnerable to false claim attacks. We argue that the vulnerability stems primarily from the ``untargeted'' characteristic of the fingerprinting methods. The untargeted characteristic leads to a large fingerprint space that can accommodate transferable adversarial fingerprints. In this paper, we propose FIT-Print, a targeted model fingerprinting framework to mitigate false claim attacks. Our main insight is that although it is tough to find the space that can only transfer among reused models, we can turn the fingerprint into a target one to restrict the fingerprinting space and reduce the adversarial transferability of fingerprints.

Given a mapping function $f(\cdot)$ and a target fingerprint $\bm{F}$, our goal is to make the fingerprint vector $\bm{v}=f(M_s(\bm{x}))$ close to $\bm{F}$. Accordingly, the definition of our FIT-Print can be defined as follows.

\vspace{0.3em}
\begin{definition}
\label{prop:fitprint}
    Let $M_s$ be the suspicious model. If for a specific mapping function $f(\cdot)$ and testing sample $\bm{x} \in \mathcal{X}_T$ ($\mathcal{X}_T$ is the set of testing samples), we have
    
    \begin{equation}
    \label{eq:target}
        \frac{1}{|\mathcal{X}_T|}\sum_{\bm{x}\in\mathcal{X}_T}{\tt dist}(f[M_s(\bm{x})], \bm{F})\leq \tau,
    \end{equation}
    where $\tau$ is a small threshold and ${\tt dist}(\cdot, \cdot)$ is a distance function, the suspicious model $M_s$ can be asserted as reused from the owner of the fingerprint $\bm{F}$.
\end{definition}

In FIT-Print, we assume that the target fingerprint $\bm{F}\in \{-1, 1\}^k$ is a binary vector consisting of $-1$ or $1$, and we can get the output logits of $M_s(\bm{x})$. The discussion on the label-only scenario, where we can only get the Top-1 label, can be found in Section~\ref{sec:label-only}. We assume that the target fingerprint needs to be registered with a third-party institution. Specifically, in the registration stage, a model developer utilizing FIT-Print registers three items to the trustworthy verifier with a timestamp: the source model ($M_o$), the optimized testing samples ($\mathcal{X}_T$), and the pre-defined target fingerprint ($\bm{F}$). Unlike conventional untargeted methods that passively register arbitrary output vectors as fingerprints, FIT-Print actively registers a semantic, specific target fingerprint (\eg, a logo) that the optimized samples are mathematically constrained to produce.
%For instance, $\bm{F}$ can be the ID number or the logo of the model developer. The signature $\bm{F}$ can be transformed into a $k$-bit string by assigning $0$ to $-1$. 

Generally, as shown in Fig.~\ref{fig:pipeline}, FIT-Print can be divided into two stages: testing sample extraction and ownership verification. %The technical details are described as follows.

\subsection{Testing Sample Extraction}
\label{sec:extraction}

In the testing sample extraction stage, we aim to find the optimal testing sample set $\mathcal{X}_T$ to make any reused models satisfy Eq.~(\ref{eq:target}) in Definition~\ref{prop:fitprint}. Therefore, in FIT-Print, we first initialize the testing samples $\mathcal{X}_T$ and the corresponding perturbations $\mathcal{R}$. We denote the $i$-th element in $\mathcal{X}_T$ and $\mathcal{R}$ as $\bm{x}_i$ and $\bm{r}_i$ respectively. The element $\bm{x}_i$ is set to an initial value $\bm{x}_i^0$, and we can initialize the testing samples to any images.
%and the perturbation $\bm{r}_i \in \mathcal{R}$ is set to zero. 
The testing samples in $\mathcal{X}_T$ can be constructed by adding the perturbations to the initial values, \ie, $\bm{x}_i=\bm{x}_i^0 + \bm{r}_i$. After that, we need to optimize the perturbations $\mathcal{R}$ to make the fingerprint vector $\bm{v}$ close to the target fingerprint $\bm{F}$. We can define the testing sample extraction as an optimization problem, which can be formalized as follows.

\begin{equation}
\label{eq:loss}
    \min_{\mathcal{R}=\{\bm{r}_1, ..., \bm{r}_{|\mathcal{R}|}\}}\frac{1}{|\mathcal{X}_T|}\sum_{i=1}^{|\mathcal{X}_T|}[\mathcal{L}(f(M_o(\bm{x}_i^0+\bm{r}_i), \bm{F})+\lambda\cdot\|\bm{r}_i\|_2],
\end{equation}
where $\|\cdot\|_2$ calculates the $\ell_2$-norm. 
% The optimization loss function in Eq.~(\ref{eq:loss}) can be regarded as two terms. 
The first term in Eq.~(\ref{eq:loss}) quantifies the dissimilarity between the output fingerprint vector $\bm{v}$ and the target fingerprint $\bm{F}$. 
%Optimizing the first term can make the metric vector $\bm{v}$ similar to the target fingerprint $\bm{F}$. 
The second term regularizes the extent of the perturbations $\mathcal{R}$. We utilize the hinge-like loss~\cite{fan2019rethinking} as $\mathcal{L}(\cdot)$, as follows.
%. The hinge-like loss is shown as follows.
\begin{equation}
\label{eq:hinge-like}
    \mathcal{L}(\bm{v}, \bm{F})=\sum_{i=1}^k \max(0, \varepsilon - \bm{v}_i\cdot\bm{F}_i).
\end{equation}

In Eq.~(\ref{eq:hinge-like}), $\bm{v}$ is the fingerprint vector, where $\bm{v}_i=f[M_s(\bm{x}_i)]$, and $\varepsilon$ is the control parameter. $\bm{F}_i$ is the $i$-th element in $\bm{F}$. Optimizing Eq.~(\ref{eq:hinge-like}) can make the signs of the corresponding elements in $\bm{v}$ and $\bm{F}$ the same. Moreover, inspired by the insight of \cite{lukas2021DeepNeuralNetwork}, we craft some augmented models by applying model reuse techniques (\eg, fine-tuning, pruning, or transfer learning) and exploit them to extract the fingerprint to improve the conferrability of FIT-Print. The set of augmented models is denoted as $\mathcal{M}$. The loss function with augmented models can be defined as Eq.~(\ref{eq:augu}).
\begin{equation}
    \label{eq:augu}
    \min_{\mathcal{R}=\{\bm{r}_1, ..., \bm{r}_{|\mathcal{R}|}\}}\frac{1}{|\mathcal{M}|\cdot|\mathcal{X}_T|}\sum_{M\in\mathcal{M}}\sum_{i=1}^{|\mathcal{X}_T|}[\mathcal{L}(\bm{v}, \bm{F})+\lambda\cdot\|\bm{r}_i\|_2].
\end{equation}
% \begin{equation}
%     \label{eq:augu}
%     \min_{\mathcal{R}=\{\bm{r}_1, ..., \bm{r}_{|\mathcal{R}|}\}}\frac{1}{|\mathcal{M}|\cdot|\mathcal{X}_T|}\sum_{M\in\mathcal{M}}\sum_{i=1}^{|\mathcal{X}_T|}[\mathcal{L}(f(M(\bm{x}_i^0+\bm{r}_i), \bm{F})+\lambda\cdot\|\bm{r}_i\|_2].
% \end{equation}

By optimizing Eq.~(\ref{eq:augu}), we can get the optimal testing samples that are conferrable to reused models, and the model developer can afterward utilize them to verify the ownership.

\subsection{Ownership Verification}
\label{sec:ov}

In the ownership verification stage, given a suspicious model $M_s$, FIT-Print examines whether the suspicious model $M_s$ is reused from the source model $M_o$ by justifying whether $M_s$ satisfies Eq.~(\ref{eq:target}). Specifically, we first calculate the fingerprint vector $\tilde{\bm{v}}$ of the suspicious model $M_s$ using the extracted testing samples in $\mathcal{X}_T$. Each element $\tilde{\bm{v}}_i = f(M_s(\bm{x}_i^0 + \bm{r}_i))$.
% \begin{equation}
%     \label{eq:extract}
%     \tilde{\bm{v}}_i = f(M_s(\bm{x}_i^0 + \bm{r}_i)), i=1, 2, ..., |\mathcal{X}_T|.
% \end{equation}
Since optimizing Eq.~(\ref{eq:augu}) makes the signs of the fingerprint vector $\tilde{\bm{v}}$ represent the fingerprint $\tilde{\bm{F}}$ of the model, we need to transform $\tilde{\bm{v}}$ into a binary vector by applying the sign function ${\tt sign}(\cdot)$ to get $\tilde{\bm{F}}$, as Eq.~(\ref{eq:binarize}).

\begin{equation}
    \label{eq:binarize}
    \tilde{\bm{F}}_i={\tt sign}(\tilde{\bm{v}}_i)=\left\{
    \begin{aligned}
        1,&\enspace \tilde{\bm{v}}_i \geq 0 \\
        -1,&\enspace \tilde{\bm{v}}_i < 0
    \end{aligned}
    \right..
\end{equation}

Subsequently, we leverage the bit error rate (BER) as the distance function ${\tt dist}(\cdot)$ in Eq.~(\ref{eq:target}) and the BER is the distance between the extracted and the target fingerprints, as follows.
\begin{equation}
    \label{eq:ber}
    {\tt BER} = \frac{1}{k}\sum_{i=1}^k\mathbb{I}\{\tilde{\bm{F}}_i\neq \bm{F}_i\},
\end{equation}
where $k$ is the length of the fingerprint and $\mathbb{I}\{\cdot\}$ is the indicator function defined as Eq. \eqref{eq:indicator}. 
\begin{equation}
    \label{eq:indicator}
    \mathbb{I}\{A\} = \left\{\begin{aligned}
        &0, A = \text{False} \\
        &1, A = \text{True}
    \end{aligned}
    \right..
\end{equation}
As Definition~\ref{prop:fitprint}, if the BER is lower than the threshold $\tau$, the suspicious model $M_s$ can be asserted as a reused model. For choosing the threshold $\tau$ to reduce false alarms and resist false claim attacks, we have Proposition~\ref{theorem:threshold}.

% \vspace{0.3em}
\begin{proposition}
    \label{theorem:threshold}
    Given the security parameter $\kappa$ and the fingerprint $\bm{F}\in\{-1, 1\}^k$, if $\tau$ satisfies that
    \begin{equation}
    \label{eq:binorm}
        \sum_{d=0}^{\lfloor \tau k \rfloor}\tbinom{k}{d} (\frac{1}{2})^k\leq \kappa,
    \end{equation}
    where $\tbinom{k}{d}=k!/[d!(k-d)!]$, the probability of a false alarm (\ie, the BER is less than $\tau$ with random testing samples) is bounded by $\kappa$.
\end{proposition}
%We propose an analysis on how to choose the threshold $\tau$ in the appendix.
The proof of Proposition~\ref{theorem:threshold} can be found in the appendix. We also conduct an empirical evaluation on the resistance of FIT-Print against adaptive false claim attacks in Section~\ref{sec:resist}.

\subsection{Designing the Mapping Function in FIT-Print}
\label{sec:implementation}

In Section~\ref{sec:extraction}-\ref{sec:ov}, we introduced the main pipeline and paradigm of our FIT-Print. The key to implementing FIT-Print is to design the mapping function $f(\cdot)$. We hereby illustrate how to leverage the paradigm of FIT-Print and design two targeted model fingerprinting methods, including FIT-ModelDiff and FIT-LIME, as the representatives of bit-wise and list-wise methods, respectively.

\subsubsection{FIT-ModelDiff}
\label{sec:fit-modeldiff}

FIT-ModelDiff is a bit-wise fingerprinting method that extracts the fingerprint bit by bit. The main insight of FIT-ModelDiff is to compare the distance between the output logits of perturbed samples $\bm{x}_i^0+\bm{r}_i$ and benign samples $\bm{x}_i^0$. The vector of the distances is called the decision distance vector (DDV). Given the suspicious model $M_s$, DDV can be calculated as follows:

\begin{equation}
\label{eq:ddv}
\begin{aligned}
    {\tt DDV}_i&={\tt cos\_sim}(M_s(\bm{x}_i^0+\bm{r}_i), M_s(\bm{x}_i^0)) \\
    &=\frac{M_s(\bm{x}_i^0+\bm{r}_i)\cdot M_s(\bm{x}_i^0)}{\|M_s(\bm{x}_i^0+\bm{r}_i)\|\cdot \|M_s(\bm{x}_i^0)\|}.
\end{aligned}
\end{equation}
 
${\tt DDV}_i$ represents the $i$-th element in the DDV and ${\tt cos\_sim}(\cdot, \cdot)$ is the cosine similarity function. Since the output logits after softmax are always positive, the range of the DDV is $[0, 1]$. As proposed in Section~\ref{sec:extraction}, we aim to make the sign of the fingerprint vector $\bm{v}$ the same as the target fingerprint $\bm{F}$. Therefore, to achieve this goal, we need to subtract a factor from DDV to make the range of $\bm{v}$ including both positive and negative values, as Eq.~(\ref{eq:ddvm}).

\begin{equation}
    \label{eq:ddvm}
    \begin{aligned}
    \bm{v}_i&=f(M_s(\bm{x}_i^0+\bm{r}_i), M_s(\bm{x}_i^0))={\tt DDV}_i - {\tt cos}(\alpha)\\
    &=\frac{M_s(\bm{x}_i^0+\bm{r}_i)\cdot M_s(\bm{x}_i^0)}{\|M_s(\bm{x}_i^0+\bm{r}_i)\|\cdot \|M_s(\bm{x}_i^0)\|}-{\tt cos}(\alpha),
    \end{aligned}
\end{equation}
where ${\tt cos}(\cdot)$ is the cosine function and $\alpha$ is the bias parameter. The final fingerprint vector $\bm{v}$ can be used for testing sample extraction or ownership verification.

\subsubsection{FIT-LIME}
\label{sec:fit-lime}

FIT-LIME is a list-wise method that extracts the fingerprint as a whole list. FIT-LIME implements the mapping function $f(\cdot)$ via a popular feature attribution algorithm, local interpretable model-agnostic explanation (LIME)~\cite{ribeiro2016should}. LIME outputs a real-value importance score for each feature in the input sample $\bm{x}$. 
%An existing fingerprinting method, Zest~\cite{jia2022ZestLIME}, utilizes primitive LIME to compare different models. However, primitive LIME has two drawbacks in ownership verification: \textbf{(1)} LIME first clusters the pixels into several groups called superpixels via Quickshift~\cite{vedaldi2008quick}. The clustering algorithm is time-consuming and these superpixels are irregular and unordered, making it hard to transform them into a bit string. \textbf{(2)} Primitive LIME depends on the label of the input to calculate the importance scores, which is not applicable when the suspicious model has different predicted classes from the source model. Therefore, 
We enhance the LIME algorithm and develop FIT-LIME to better cater to the needs of ownership verification. The details of FIT-LIME are elaborated as follows.

\begin{table*}[t]
% \vspace{-0.8em}
    \tabcolsep=2.5mm
    \renewcommand{\arraystretch}{1.1}
    \centering
    \caption{Successful ownership verification rates of different model fingerprinting methods. ``\#Models'' denotes the number of reused models, and the ``N/A'' indicates that the method cannot be applied to detect this type of model reuse technique. In particular, we mark failed cases (\ie, $<80\%$ or ``N/A'') in red. Moreover, the BERs of FIT-Print are all 0.0\%.}
    \label{tab:baseline}
    \scalebox{0.85}{
    \begin{tabular}{cc|cc|ccc|c|cc}
    \hline
    \hline
    % \toprule
        \multirow{2}{*}{Reuse Task$\downarrow$} & \multirow{2}{*}{\#Models$\downarrow$} & \multicolumn{2}{c|}{AE-based} & \multicolumn{3}{c|}{Testing-based} & White-box & \multicolumn{2}{c}{FIT-Print}\\
         & & IPGuard & MetaV & ModelDiff & Zest & SAC & ModelGiF & FIT-ModelDiff & FIT-LIME \\
         \hline
         Copying & 4 & 100\% & 100\%  & 100\% & 100\% & 100\% & 100\% & 100\% & 100\%\\
         Fine-tuning & 12&  100\% & 100\%  & 100\% & 100\% & 100\% & 100\% & 100\% & 100\%\\
         Pruning & 12& 100\% & 100\% &  100\%  & 91.67\% & 100\% & 100\% & 100\% & 100\%\\
         Extraction & 8& \red{50\%} & 87.5\% &  \red{50\%}  & \red{25\%} & 100\% & 100\% & 100\% & 100\%\\
         Transfer & 12& \red{N/A} & \red{N/A}  & 100\% & \red{N/A} & \red{0\%} & 100\% & 100\% & 100\%\\
         \hline
         Independent & 144 & 30.6\% & 4.8\% & 4.0\% & 7.6\% & 39.6\% & 0.0\% & 0.0\% & 0.0\%\\
         \hline
         \hline
         % \bottomrule
    \end{tabular}
    }
    \vspace{-1em}
\end{table*}

The first step of FIT-LIME is to generate $c$ samples that are neighboring to the input image $\bm{x}$. We also gather the adjacent pixels in the image into a superpixel. We uniformly segment the input space into $k$ superpixels, where $k$ is the length of the targeted fingerprint. Assuming that $k=\mu\times \nu$, the image can be divided into $\mu$ rows and $\nu$ columns. 
Each superpixel represents a rectangular region that has $\lceil w/\mu \rceil \times \lceil h/\nu \rceil$ pixels in the image. 
Then, we randomly generate $c$ masks where each mask is a $k$-dimension binary vector, constituting $\bm{A}\in\{0, 1\}^{c\times k}$. Each element in each row of the matrix $\bm{A}$ corresponds to a superpixel in the image $\bm{x}$. After that, we exploit the binary matrix to mask the image $\bm{x}$ and generate the masked examples $\mathcal{X}^m$. If the element in the $i$-th row of the mask $\bm{A}$ is 1, the corresponding superpixel preserves its original value. Otherwise, the superpixel is aligned with 0. Each row of the mask can generate a masked sample, and the $c$ masked samples constitute the masked sample set $\mathcal{X}^m$.

The second step is to evaluate the output of the masked samples $\mathcal{X}^m$ on the suspicious model $M_s$. Different from primitive LIME, we utilize the entropy of the outputs so that it no longer depends on the label of $\bm{x}$. The intuition is that if the important features are masked, the prediction entropy will significantly increase. Following this insight, we calculate the following equation in this step.

\begin{equation}
    \label{eq:entropy}
    \bm{p}_i=H[M_s(\mathcal{X}^m_i)],
\end{equation}

In Eq. (\ref{eq:entropy}), $H(\cdot)$ calculates the entropy, $\bm{p}_i$ is the $i$-th element in $\bm{p}$, and $\mathcal{X}^m_i)$ is the $i$-th masked sample.

After that, the final step is to fit a linear surrogate model to approximate the relationship between the binary masks $\bm{A}$ and the corresponding prediction entropies $\bm{p}$. Specifically, following \cite{shao2024explanation}, we optimize an ordinary least squares objective to find the feature importance vector $v$ that minimizes the approximation error:
\begin{equation}
\label{eq:linear_loss}
    \min_{\bm{v}}\|\bm{A}\bm{v}-\bm{p}\|_2^2.
\end{equation}

Eq.~\eqref{eq:linear_loss} can be solved using the well-known Normal Equation. This yields the importance score vector $\bm{v}$ as formulated in Eq.~\eqref{eq:linear}:
\begin{equation}
    \label{eq:linear}
    \bm{v}=(\bm{A}^T\bm{A})^{-1}\bm{A}^T\bm{p}.
\end{equation}

This importance score vector $\bm{v}$ is then utilized as the fingerprint vector in both the testing sample extraction and ownership verification stages.

\section{Experiments}
\label{sec:exp}

% TODO:这里总结一下我们正文和附录做的实验
In this section, we evaluate the effectiveness, conferrability, and resistance to the false claim attack of our FIT-Print methods. We also include ablation studies on the hyperparameters, different targets, initializations, and different numbers of augmented models in FIT-Print. 
% We also discuss applying FIT-Print in the label-only scenario and to other models and datasets. The analysis of the overhead of FIT-ModelDiff and FIT-LIME can be found in Appendix.

\begin{figure*}
    \centering
    \includegraphics[width=0.80\linewidth]{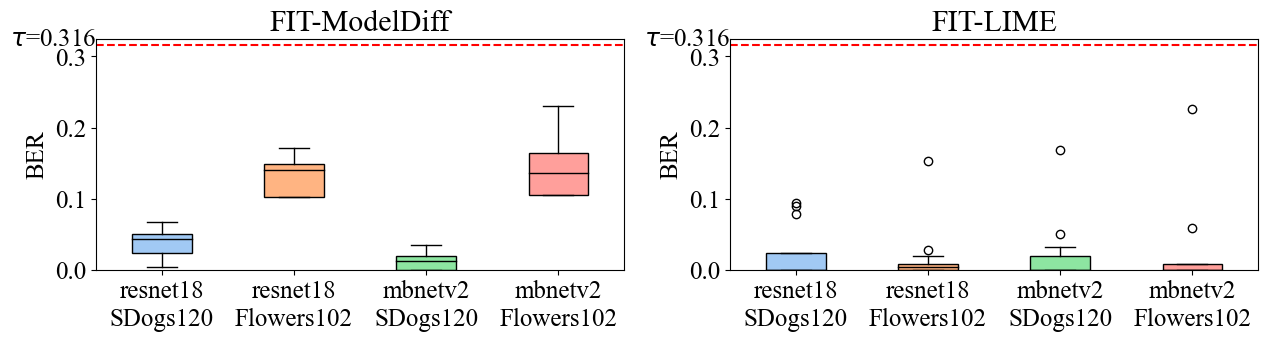}
    % \vspace{-1em}
    \caption{The BERs of different source models and their reused models with FIT-ModelDiff and FIT-LIME. The BERs are all less than the threshold $\tau$ marked with a red dashed line, indicating that FIT-ModelDiff and FIT-LIME can successfully recognize the reused models.}
    \label{fig:berbox}
    \vspace{-1em}
\end{figure*}

\subsection{Experimental Settings}
\label{sec:settings}

\partitle{Models and Datasets} Following prior works~\cite{li2021ModelDiff, jia2022ZestLIME}, we utilize two widely-used convolutional neural network (CNN) architectures, MobileNetV2~\cite{sandler2018mobilenetv2} (mbnetv2 for short) and ResNet18~\cite{he2016deep}, in our experiments. We train MobileNetV2 and ResNet18 using two different datasets, Oxford Flowers 102 (Flowers102)~\cite{nilsback2008automated} and Stanford Dogs 120 (SDogs120)~\cite{khosla2011novel}, in total 4 source models. Flowers102 contains RGB images of 102 different flowers, while SDogs120 has images of 120 different breeds of dogs. We primarily focus on image classification models in our experiments. In particular, we provide a case study about implementing FIT-Print to text generation models in Section~\ref{sec:case}.
%Flowers102 contains RGB images of 102 different flowers, while SDogs120 has images of 120 different breeds of dogs. 
We utilize the data from ImageNet~\cite{deng2009imagenet} as the default initial images of the testing samples $\mathcal{X}_T$. 
%We also conduct experiments on different initializations in the appendix.
We also utilize the models pre-trained on ImageNet~\cite{deng2009imagenet} as the models trained by other parties, \ie, negative models.

\partitle{Model Reuse Techniques} We evaluate FIT-Print against the following five categories of model reuse techniques, including copying, fine-tuning, pruning, model extraction, and transfer learning. We further consider different implementations of these model reuse techniques in various settings and scenarios. 
%In our experiments, we have four source models (MobileNetV2 and ResNet18 trained on Flowers102 and SDogs120 datasets respectively). 
For each source model, we train and craft three fine-tuning models, three pruning models, two extraction models, and three transfer learning models. These 12 models constitute the set of reused models. When experimenting on one source model, the other 36 models that are reused from other source models are treated as independent models. %More details can be found in Appendix~\ref{sec:reusedetails}.

\begin{figure*}[t]
% \vspace{-0.4em}
    \centering
    \includegraphics[width=0.80\linewidth]{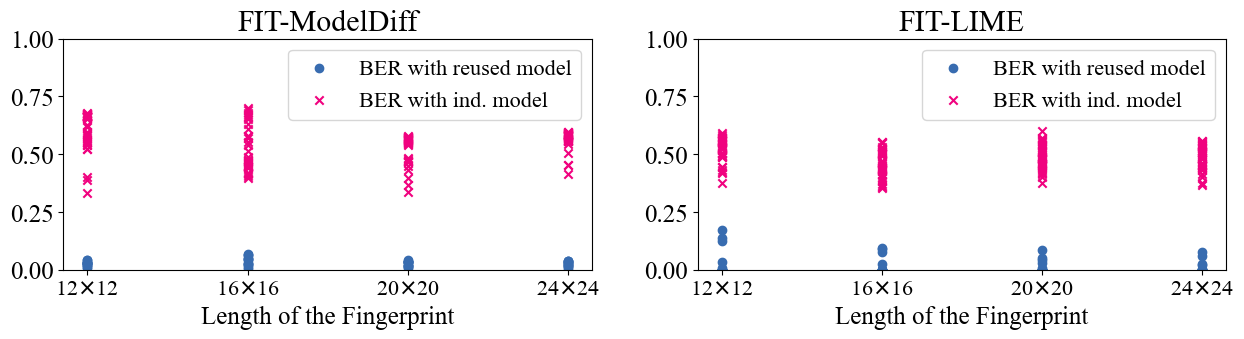}
    % \vspace{-1em}
    \caption{The BERs of the reused models and independent models with different lengths of fingerprint. As the length increases, the BERs with reused and independent models become more concentrated.}
    \label{fig:fp-length}
    % \vspace{-0.5em}
\end{figure*}

\begin{table*}[t]
    \renewcommand{\arraystretch}{1.1}
    \centering
    \caption{The average distances of reused models (Avg. Reused Model Dist.) and independent models (Avg. Ind. Model Dist.) and the $\ell_2$-norm of the perturbations $\mathcal{R}$ ($\ell_2$-norm of Pert.) with different $\lambda$.}
    \label{tab:lambda}
    \scalebox{0.82}{
    \begin{tabular}{c|ccccc|ccccc}
    \hline
    \hline
    % \toprule
    Method$\rightarrow$ & \multicolumn{5}{c|}{FIT-ModelDiff} & \multicolumn{5}{c}{FIT-LIME} \\
    Metric $\downarrow$ $\lambda \rightarrow$  &  0.0 & 1.0 & 5.0 & 10.0 & 100.0 & 0.0 & 0.5 & 1.0 & 2.0 & 5.0\\
    \hline
    Avg. Reused Model Dist. & 0.024 & 0.038 & 0.030 & 0.032 & 0.029 & 0.029 & 0.029 & 0.034 & 0.036 & 0.047\\
    Avg. Ind. Model Dist. & 0.568 & 0.570 & 0.561 & 0.566 & 0.577 & 0.505 & 0.505 & 0.510 & 0.506 & 0.512\\
    $\ell_2$-norm of Pert. & 0.007 & 0.007 & 0.007 & 0.007 & 0.006 & 0.020 & 0.019 & 0.018 & 0.017 &  0.014\\
         \hline
         \hline
         % \bottomrule
    \end{tabular}
    }
    \vspace{-1em}
\end{table*}

\partitle{Baseline Methods} 
We consider both AE-based and testing-based fingerprinting methods. 
For the former, we implement two typical methods, IPGuard~\cite{cao2021ipguard} and MetaV~\cite{pan2022metav}. While for the latter, we take three different methods,  ModelDiff~\cite{li2021ModelDiff}, Zest~\cite{jia2022ZestLIME}, and SAC~\cite{guan2022AreYouStealing} as the baseline methods. We also include a state-of-the-art (SOTA) white-box model fingerprinting method, \ie, ModelGiF~\cite{song2023modelgif}, for reference.

\partitle{Target Fingerprint} As default, we select a logo of a file and a pen as the targeted fingerprint $\bm{F}$. We set the default length $k$ of the fingerprint $\bm{F}$ to be 256 and thus $\bm{F}$ is resized to 16$\times$16. We set the security parameter $\kappa=10^{-9}$. According to Eq.~(\ref{eq:binorm}), the threshold $\tau$ is $0.316$ in our experiments.

\partitle{Optimization Details} We utilize the stochastic gradient descent (SGD) with momentum as the optimizer. We set the initial learning rate to $1.2\times 10^{-2}$, the momentum to $0.9$, and the weight decay to $5\times 10^{-4}$. We apply a cosine annealing schedule~\cite{loshchilov2016sgdr} to reduce the learning rate gradually to a minimum of $4\times 10^{-3}$. Following~\cite{shao2024explanation}, we set the control parameter $\varepsilon$ in the hinge-like loss in Eq.~(\ref{eq:hinge-like}) to $0.01$. We optimize the perturbations on the testing samples for 300 epochs. In Eq.~(\ref{eq:ddvm}) of FIT-ModelDiff, we set the bias parameter $\alpha$ to be $\pi / 8$.

\partitle{Computational Resources} All our experiments are implemented with at most 8 RTX 3090 GPUs.

\subsection{Evaluation on Effectiveness and Conferrability}
\label{sec:eff}

%To evaluate the effectiveness and conferrability of FIT-Print, we extract the testing samples and validate the fingerprint on the reused models. 
%We utilize 7 reused models as augmentation models in Eq.~(\ref{eq:augu}) to extract the testing samples and evaluate the effectiveness and conferrability of FIT-Print. 
Table~\ref{tab:baseline} illustrates the percentage of successfully identified reused models (ownership verification rate). Both FIT-ModelDiff and FIT-LIME can recognize the reused models under five reuse techniques with $100\%$ ownership verification rates, which outperform existing fingerprinting methods and perform on par with the SOTA white-box method, ModelGiF. Also, FIT-ModelDiff and FIT-LIME achieve $0.0\%$ ownership verification rates on the independent models, indicating that our methods do not lead to false alarms. Fig.~\ref{fig:berbox} illustrates the BERs of the reused models, which are all less than the threshold $\tau$ with a maximum of $0.227$. The results validate the effectiveness and conferrability of FIT-Print. 

\subsection{Ablation Study}
\label{sec:ablation}

\subsubsection{Effects of the Length of the Fingerprint}
\label{sec:fp-length}

In this experiment, we investigate the impact of varying lengths of the fingerprint $\bm{F}$.  In addition to the default length of $256=16\times 16$, we set the length to be $12\times 12$, $20\times 20$, and $24\times 24$. The results illustrated in Fig.~\ref{fig:fp-length} indicate that both FIT-ModelDiff and FIT-LIME can recognize the reused models and the independent models with different lengths of fingerprints. Moreover, as the length of $\bm{F}$ increases, the BERs of both reused and independent models are more concentrated, signifying that a larger fingerprint length can reduce the probability of outliers and have better security.

\subsubsection{Effects of the $\ell_2$-norm Coefficient}
\label{sec:lambda}

% 这里一张图，展示不同lambda的扰动后的图。（放附录）
% 还有一张表，分不同lambda，展示BER区间，和Average扰动
$\lambda$ is the coefficient of the scale of the perturbations in the loss function Eq.~(\ref{eq:augu}). In this experiment, we study the effect of $\lambda$ on FIT-Print and adopt FIT-ModelDiff and FIT-LIME with five different $\lambda$. From Table~\ref{tab:lambda}, since the scale of the perturbations in FIT-ModelDiff is quite small, varying $\lambda$ does not significantly affect the perturbations as well as the distances with reused models. While in FIT-LIME, a larger $\lambda$ can lead to a smaller perturbation. The $\ell_2$-norm of the perturbations reduces from 0.020 to 0.014. In the meantime, the average distances with reused models become larger. Our experiments also suggest that the effect of $\lambda$ on the distances with independent models is not significant. 
The visualization of the perturbed testing samples with different $\lambda$ can be found in Fig.~\ref{fig:visual-lambda} and it also confirms our above conclusion.

\begin{figure}[t]
    \centering
    \includegraphics[width=1.0\linewidth]{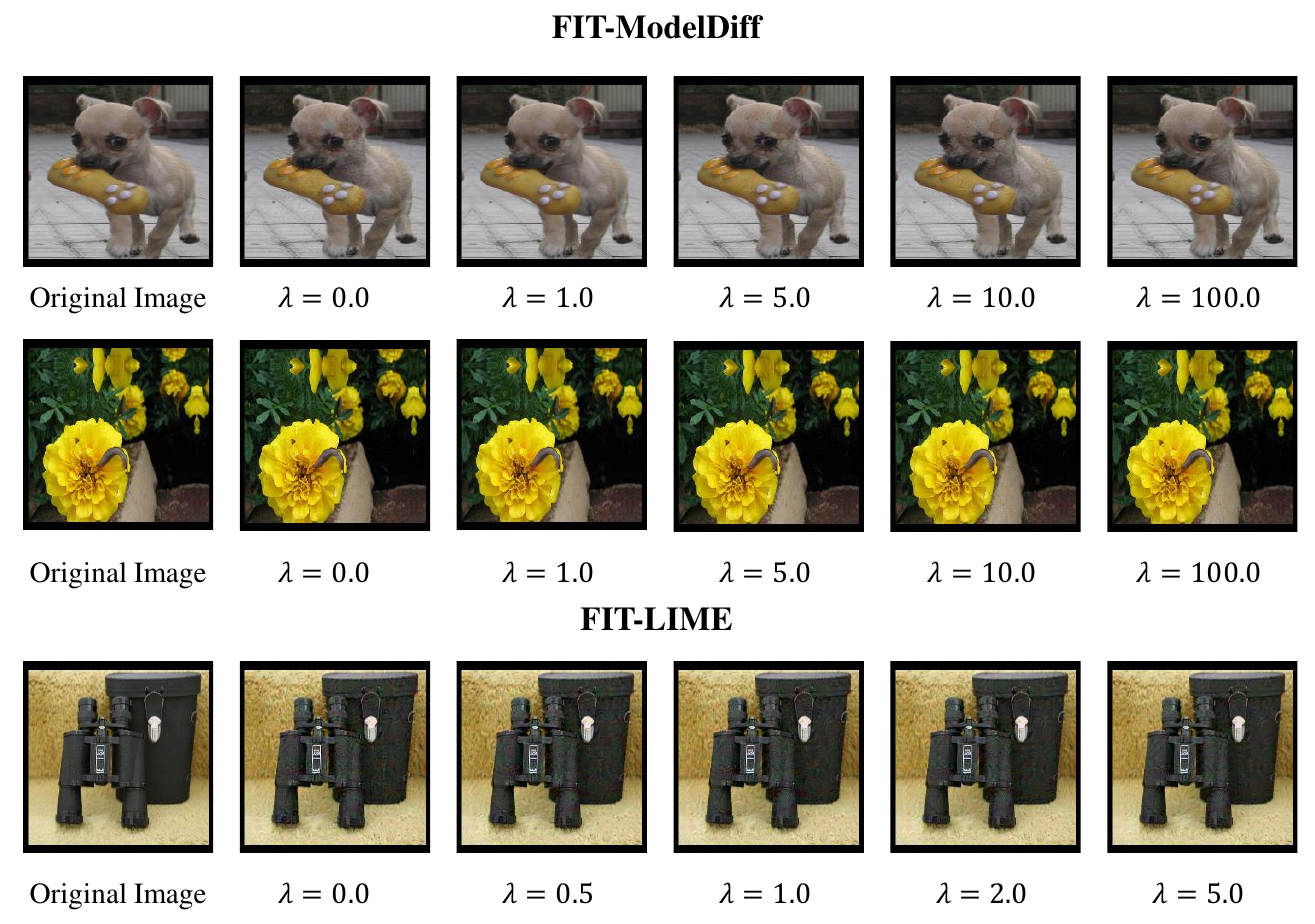}
    \caption{The visualization of the original images and the perturbed testing samples with different $\lambda$.}
    \label{fig:visual-lambda}
    % \vspace{-5pt}
\end{figure}

\begin{figure}[t]
    \centering
    \subfloat[Target-file]{
    \includegraphics[width=0.21\linewidth]{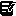}
    }\hfill
    \subfloat[Target-tick]{
    \includegraphics[width=0.21\linewidth]{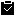}
    }\hfill
    \subfloat[Target-noise]{
    \includegraphics[width=0.21\linewidth]{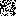}
    }\hfill
    \caption{The visualization of the targeted fingerprints. We utilize three different target fingerprints in our experiments. The ``Target-file'' fingerprint is used in our main experiments.}
    \label{fig:fingerprint}
    \vspace{-1em}
\end{figure}

\begin{figure*}[t]
    \centering
    \includegraphics[width=0.80\linewidth]{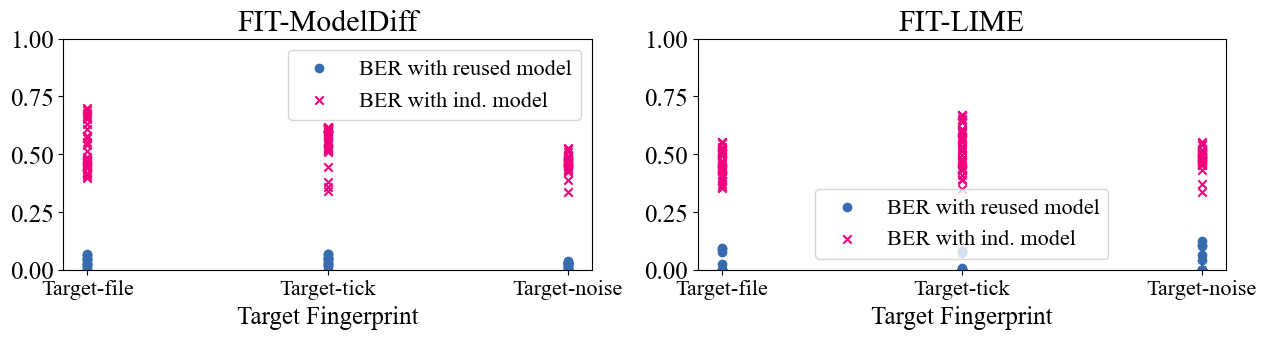}
    \caption{The BERs of the reused models and independent models with different target fingerprints. Regardless of the target fingerprints, our FIT-ModelDiff and FIT-LIME have the capability to distinguish the reused models and the independent models.}
    \label{fig:ber-target}
    \vspace{-1em}
\end{figure*}

\begin{figure*}[t]
    \centering
    \includegraphics[width=0.80\linewidth]{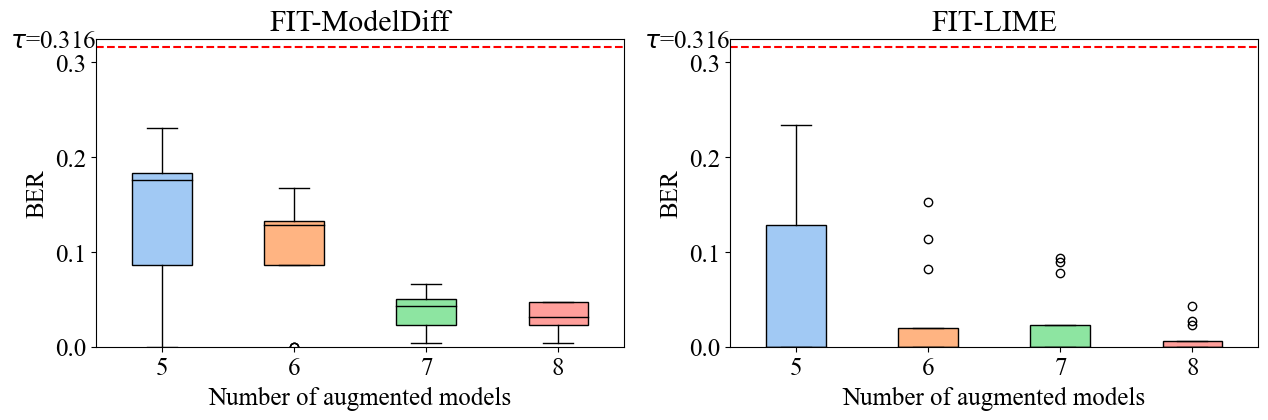}
    \caption{The BERs of the reused models with different numbers of augmented models. As shown in this figure, as we expected, using more models for augmentations can have lower BERs.}
    \label{fig:aug-models}
    \vspace{-1em}
\end{figure*}

\subsubsection{Effects of Different Target Fingerprints}

% Target Fingerprint的图

% \noindent\textbf{How to Choose the Target Fingerprint $F$.} We briefly introduce how to choose the target fingerprint. In our method, the targeted fingerprint is a bit string representing the identity of the model developer and needs to be registered to the trustworthy verifier. For instance, the company's logo or personal identity number can be used as a targeted fingerprint. We note that the choice of the fingerprint does not affect the model performance. This is because model fingerprinting does not alter the models' parameters and has no impact on the model performance. This is a key advantage of fingerprinting.

%\noindent\textbf{Experiments with Different Target Fingerprints.} 
%In the main experiments of our paper, we utilize an image of a `file' and a `pen' as the target fingerprint $\bm{F}$.
We hereby explore how the properties of the target fingerprint (\eg, spatial pattern and complexity) affect the collision probability (\ie, false alarms on independent models) and overall security. Specifically, we choose three distinct target fingerprints representing different levels of complexity: \textbf{(1)} a ``tick'' image (simple spatial pattern, moderate complexity), \textbf{(2)} an image of a ``file'' and a ``pen'' (high complexity), and \textbf{(3)} a random noise image (no spatial pattern, maximum entropy). The visualization of the three fingerprints (resized to $16 \times 16$) is shown in Fig.~\ref{fig:fingerprint}.

The experimental results are shown in Fig.~\ref{fig:ber-target}. From Fig.~\ref{fig:ber-target}, we can find that regardless of the signature's pattern and complexity, all the BERs of reused models are lower, and the BERs of independent models are larger than the threshold. This demonstrates that our FIT-Print is effective with different target fingerprints.

\begin{figure*}[t!]
    \centering
    \includegraphics[width=0.80\linewidth]{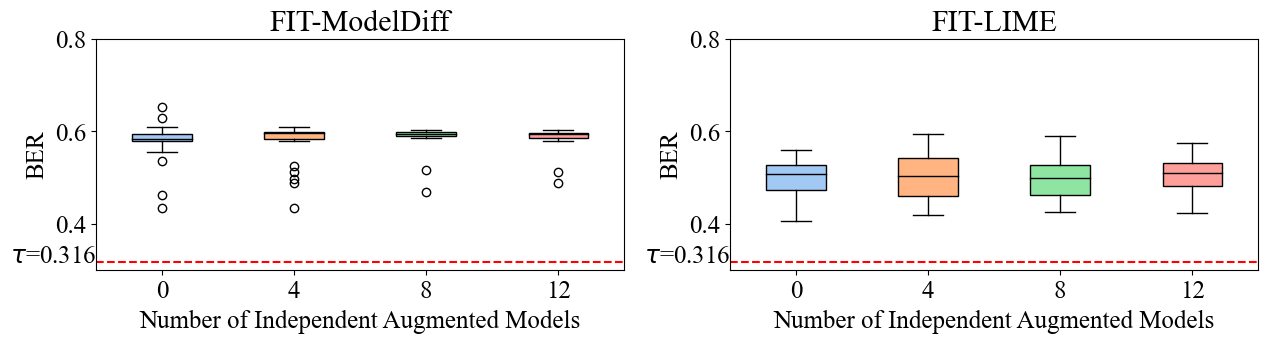}
    % \vspace{-1em}
    \caption{The BERs of the independent models while conducting adaptive false claim attacks using different numbers of independent models as augmented models. The BERs are all larger than the threshold $\tau$.}
    \label{fig:negative}
    % \vspace{-0.5em}
\end{figure*}

\subsubsection{Effects of Different Numbers of Augmented Models}
\label{sec:aug-model}

In the testing sample extraction stage, FIT-Print utilizes the reused models as augmented models to enhance the conferrability of the fingerprint. In this section, we study leveraging different numbers of augmented models to extract the testing samples and test whether FIT-Print maintains a satisfactory conferrability. Fig.~\ref{fig:aug-models} depicts the BERs of the reused models with different numbers of augmented models. We set the number to 5, 6, 7, 8. From Fig.~\ref{fig:aug-models}, we can find that as the number of augmented models increases, the BERs on reused models become smaller and more concentrated, which means using more reused models as augmented models can enhance the conferability of FIT-Print. In addition, while using only 5 reused models as augmented models, all the BERs are smaller than the threshold $\tau$, which signifies the conferrability of FIT-Print. 

\subsubsection{Effects of the Bias Parameter $\alpha$}

In this section, we study the influence of the bias parameter $\alpha$ (as defined in Eq.~\eqref{eq:ddvm}) on the performance of FIT-ModelDiff. We set $\alpha$ to $\frac{1}{16}\pi, \frac{1}{8}\pi, \frac{1}{4}\pi$, and $\frac{1}{3}\pi$. The results in Table~\ref{tab:alpha} show that larger $\cos\alpha$ leads to lower BERs, and thus may cause better robustness. However, our method is generally robust to the choice of the bias parameter $\alpha$.

\begin{table}[t]
% \vspace{-0.8em}
    % \tabcolsep=1mm
    \renewcommand{\arraystretch}{1.1}
    \centering
    \caption{The average bit error rates of different $\alpha$. The results demonstrate that larger $\cos\alpha$ leads to lower BERs, but our method is generally robust to the choice of $\cos\alpha$.}
    \label{tab:alpha}
    \scalebox{0.85}{
    \begin{tabular}{c|cccc}
    % \hline
    % \hline
    \hline
    \hline
    $\alpha$ & $\frac{1}{16}\pi$ & $\frac{1}{8}\pi$ & $\frac{1}{4}\pi$ & $\frac{1}{3}\pi$ \\
    $\cos\alpha$ & 0.9808 & 0.9239 & 0.7071 & 0.5000\\
    \hline
    Avg. BER ($\downarrow$) & 0.015 & 0.033 & 0.118 & 0.217 \\
    \hline
    \hline
    \end{tabular}
    }
    \vspace{-1em}
\end{table}

\subsection{The Resistance to Adaptive False Claim Attack}
\label{sec:resist}

We hereby evaluate our FIT-Print against the adaptive false claim attack, where the adversary utilizes Eq.~(\ref{eq:augu}) to optimize the testing samples yet intentionally crafts some independent models as augmented models to enhance the transferability of the adversary's false fingerprint. We utilize the models trained on ImageNet and their corresponding reused models as the augmented models. Fig.~\ref{fig:negative} demonstrates that adding independent augmented models does not significantly enhance the transferability of the fingerprint in FIT-Print because the BERs on the independent models are nearly unchanged. However, as shown in Section~\ref{sec:falseclaim}, existing fingerprinting methods are vulnerable to false claim attacks due to their untargeted nature. In contrast, our targeted FIT-Print is significantly more resistant to false claims since targeted transferable fingerprints are much more difficult to craft (as analyzed in Section~\ref{sec:framework}). Empirically, such a phenomenon is also validated in the area of adversarial attacks~\cite{wang2023lfaa, wang2023towards} (\ie, targeted adversarial examples have lower transferability than untargeted ones). 

\subsection{The Resistance to Adaptive Fingerprint Removal Attacks}

In real-world scenarios, the model reuser usually knows which model fingerprinting method is leveraged by the model developer and can accordingly design an adaptive attack against the utilized model fingerprinting method. Generally speaking, there are two different ways to attack a model fingerprinting method~\cite{yao2023removalnet}: \textbf{(1)} fine-tuning the model (\ie, model-based attacks) or \textbf{(2)} perturbing or preprocessing the input data (\ie, input-based attacks) to obfuscate the fingerprint of the model. %In this section, we primarily focus on the former attack. The details and discussions of the latter attack can be found in Appendix~\ref{sec:input-based}.

In model-based adaptive attacks, the model reuser can fine-tune the model, attempting to remove the original fingerprint inside it. Based on the knowledge of the model reuser with the fingerprint of the model developer, we consider two different model-based adaptive attacks.

\begin{itemize}[leftmargin=*]
    \item \emph{Overwriting Attack:} In overwriting attacks, we assume that the model reuser has no knowledge of the testing samples $\mathcal{X}_T$ and the target fingerprint $\bm{F}$ utilized by the model developer. Thereby, the model reuser may independently generate their own $\hat{\mathcal{X}}_T$ and $\hat{\bm{F}}$, and then fine-tunes the model to make the outputs of $\hat{\mathcal{X}}_T$ close to the target fingerprint $\hat{\bm{F}}$. The loss function can be defined as follows.
    \begin{equation}
        \min_{M_o}\frac{1}{|\hat{\mathcal{X}}_T|}\sum_{\hat{\bm{x}} \in \hat{\mathcal{X}}_T}\mathcal{L}(f(M_o(\hat{\bm{x}}), \hat{\bm{F}}).
    \end{equation}
    \item \emph{Unlearning Attack:} In unlearning attacks, we assume that the model reuser knows the target fingerprint of the model developer, since it may be registered in a third-party institution and publicly accessible. However, the model reuser still has no knowledge of the testing samples. As such, the model reuser can construct some independent testing samples to unlearn the target fingerprint $\bm{F}$ from the model. The loss function can be defined as follows.
    \begin{equation}
        \max_{M_o}\frac{1}{|\hat{\mathcal{X}}_T|}\sum_{\hat{\bm{x}} \in \hat{\mathcal{X}}_T}\mathcal{L}(f(M_o(\hat{\bm{x}}), \bm{F}).
    \end{equation}
\end{itemize}

\begin{table}[t]
    \renewcommand{\arraystretch}{1.1}
    \caption{The BERs before and after the adaptive overwriting attack and unlearning attack. The BERs after attacks are still low enough to be identified as a reused model. Thus, FIT-Print is able to resist the overwriting attack and unlearning attack.}
    \label{tab:adaptive}
    \centering
    \scalebox{0.80}{
    \begin{tabular}{c|ccc}
        \hline
        \hline
        % \toprule
        Method &  Before Attack & After Overwriting & After Unlearning\\
        \hline
        FIT-ModelDiff & 0.047 & 0.051 & 0.149\\
        FIT-LIME & 0.000 & 0.016 & 0.016\\
        \hline
        \hline
        % \bottomrule
    \end{tabular}
    }
    % \vspace{-1em}
\end{table}

The results of the two adaptive attacks are demonstrated in Table~\ref{tab:adaptive}. The results indicate that neither attack can successfully remove the fingerprint from the model. Due to more knowledge about the fingerprint $\bm{F}$, the unlearning attack is slightly more effective than the overwriting attacks, but still not able to bypass the ownership verification, with a BER of 0.149. The experimental results show that the model reuser cannot destroy the fingerprint inside the model when having no knowledge of the testing samples. Our FIT-Print can resist both the overwriting attack and the unlearning attack.

\begin{table}[t]
    \tabcolsep=4mm
    \renewcommand{\arraystretch}{1.1}
    \caption{The performance of FIT-ModelDiff and FIT-LIME in the label-only scenario. The results show that FIT-Print can still distinguish the reused models with only top-1 labels.}
    \label{tab:label-only}
    \centering
    \scalebox{0.85}{
    \begin{tabular}{c|cc}
        \hline
        \hline
        % \toprule
        Metric$\downarrow$ Method$\rightarrow$ &  FIT-ModelDiff & FIT-LIME\\
        \hline
        Avg. BER & 0.227 & 0.135 \\
        Ownership Verification Rate & 1.000 & 1.000 \\
        \hline
        Avg. BER of Ind. Models & 0.369 & 0.399 \\
        False Positive Rate & 0.000 & 0.000 \\
        \hline
        \hline
        % \bottomrule
    \end{tabular}
    }
    % \vspace{-1em}
\end{table}

\subsection{FIT-Print in the Label-only Scenario}
\label{sec:label-only}

In this section, we investigate the effectiveness of FIT-Print in the label-only scenario. In such a scenario, the verifier can only obtain the Top-1 label instead of the logits as output. Unfortunately, detecting transfer learning models, which is one of our considered important model reuse settings, is still an open problem in model ownership verification~\cite{sun2023deep}. Transfer learning can change the task of the model and the output classes. It is hard to determine whether a model is transferred from another with only top-1 labels. As such, we do not consider transfer learning models in the following discussion.

FIT-Print can easily be extended to distinguish the reused models (except transfer learning models) with the top-1 labels. In the label-only scenario, we can construct a binary vector $\bm{b}$ to replace the original logits. Assuming that the predicted top-1 class is $a$, the $a$-th element in $\bm{b}$ is set to 1, and the other elements are 0. The other processes remain unchanged.

To verify the effectiveness, we conduct additional experiments. Table~\ref{tab:label-only} shows the average bit error rate (Avg. BER) on the reused models and the independent models (Ind. Models). We also present the ownership verification rates on the reused models and the false positive rates on the independent models. The results show that our methods are still highly effective under the label-only setting, although the average BERs of the reused models decrease in this scenario.

\section{Discussion}

\subsection{Related Work}
\label{sec:related}

As the deployment of deep learning expands, ensuring the security, privacy, and trustworthiness of AI models and systems has drawn significant attention~\cite{jiang2024lancelot, xie2025harmony}. Within the broader context of trustworthy AI, in this section, we provide a comprehensive discussion on model fingerprinting. 
% Some existing studies have been developed to compute `the functional distance' between different models. Although these works serve a different purpose from model fingerprinting for ownership verification, they share technical similarities. Therefore, we collectively refer to these works as model fingerprinting. 
In general, existing model fingerprinting can be categorized into white-box and black-box methods~\cite{sun2023deep}.

\partitle{White-box Model Fingerprinting Methods} In the white-box scenario, since the model developer can get access to the parameters of the suspicious model, a direct way to compare the models is to compare the weights (or their hash values) of the models. Some existing white-box model fingerprinting methods leveraged the path of model training~\cite{jia2021proof}, the random projection of model weights~\cite{zheng2022dnn}, or the learnable hash of the model weights~\cite{xiong2022neural}. Some recent works also explored utilizing the deep representations (\eg, gradients~\cite{song2023modelgif}) or the intermediate results~\cite{chen2022copy} of the testing samples as the fingerprint. Recently, these concepts have also been extended to large language models (LLMs)~\cite{shao2025sok} by extracting unique fingerprints from static transformer weights~\cite{zeng2024huref}, dynamic forward-pass representations~\cite{zhang2025reef}, or backward-pass gradients~\cite{wu2025tensorguard}. However, similar to traditional white-box watermarking, the strict requirement of full parameter access severely restricts their practical deployment in real-world scenarios, especially for closed-source commercial models.

\partitle{Black-box Model Fingerprinting Methods} In the black-box scenario, the model developer is assumed to have only API access to the suspicious model. Existing black-box model fingerprinting methods can be classified into adversarial-example-based (AE-based) methods~\cite{cao2021ipguard, wang2021characteristic, lukas2021DeepNeuralNetwork} and testing-based methods~\cite{li2021ModelDiff, chen2022teacher, chen2022copy} and we have presented their formulations in Section~\ref{sec:formulation}. AE-based methods craft adversarial examples to identify the decision boundary of different models~\cite{cao2021ipguard}. Lukas et al.~\cite{lukas2021DeepNeuralNetwork} proposed to craft some reused models as augmented models to improve the conferrability of the fingerprint. On the contrary, testing-based methods compare the model behavior on the testing samples at the specific mapping function. ModelDiff~\cite{li2021ModelDiff} and SAC~\cite{guan2022AreYouStealing} utilized the distances between the output logits of different input samples, while Zest~\cite{jia2022ZestLIME} took the feature attribution map output by LIME~\cite{ribeiro2016should} as the mapping function. In addition, Chen et al.~\cite{chen2022copy} proposed a series of mapping functions to calculate model similarity. Compared to AE-based methods, testing-based methods have the capability to compare models across different tasks and output formats, thereby attracting greater attention in practice. In the context of modern LLMs, black-box approaches are further categorized into untargeted and targeted paradigms~\cite{shao2025sok}. While untargeted methods analyze the statistical or semantic features of general generation responses~\cite{pasquini2025llmmap, shao2026reading}, targeted LLM fingerprinting utilizes adversarially optimized prompts to force the model to output specific, pre-defined textual signatures~\cite{gubri2024trap}. These emerging LLM studies share similar underlying philosophies with our FIT-Print paradigm in pursuing more reliable and resilient copyright auditing.

\subsection{Extending FIT-Print to Other Models and Datasets}
\label{sec:extension}

In our main experiments, we focus on image classification models. It is also technically feasible to extend our FIT-Print to models of any task. In this section, we discuss how FIT-Print can generalize to different types of models and data.

\partitle{The Extension to Other Models} 
We argue that our FIT-Print can also generalize to models with different architectures and tasks. For models with different architectures, since we do not make any assumptions about the architecture of the models and we also do not need to alter or fine-tune the model, our method can fundamentally generalize to models with other architectures (\eg, transformers~\cite{vaswani2017attention}) as well. For models with different tasks, the major difference between models with different tasks is the output format. For instance, the image generation model outputs a tensor consisting of a sequence of logits. FIT-ModelDiff calculates the cosine similarity between the outputs, and FIT-LIME calculates the average entropy of the output. Arguably, these calculation methods can be applied to any output format (\eg, 1-D vectors, 2-D matrices, or tensors). As such, our methods are naturally feasible for models with different tasks in practice.

To further demonstrate the extensibility, we conduct additional empirical evaluations on several advanced and widely adopted image classification architectures, specifically Vision Transformer (ViT)~\cite{dosovitskiy2020image}, Swin Transformer (SwinViT)~\cite{liu2021swin}, and EfficientNet~\cite{tan2019efficientnet}. The experimental results are in Table~\ref{tab:more_models}. As illustrated, both FIT-ModelDiff and FIT-LIME maintain exceptional performance with BERs $< 0.17$ across these diverse model architectures. These results compellingly indicate that FIT-Print can successfully extract targeted fingerprints and verify ownership regardless of the underlying model structure.

\begin{table}[t]
    \tabcolsep=4mm
    \renewcommand{\arraystretch}{1.2}
    \centering
    \caption{The bit error rates (BER) of applying FIT-ModelDiff and FIT-LIME to advanced image classification models.}
    \label{tab:more_models}
    \scalebox{0.88}{
    \begin{tabular}{c|ccc}
    \hline
    \hline
    % \toprule[1.5pt]
    % Method & \multicolumn{3}{c|}{FIT-ModelDiff} & \multicolumn{3}{c}{FIT-LIME} \\
    Model  & ViT & SwinViT & EfficientNet \\
    \hline
    FIT-ModelDiff & 0.117 & 0.125 & 0.164 \\
    FIT-LIME & 0.008 & 0.000 & 0.012 \\
         \hline
         \hline
         % \bottomrule[1.5pt]
    \end{tabular}
    }
    % \vspace{-0.5em}
\end{table}

% \vspace{0.3em}
\partitle{The Extension to Other (Types of) Datasets}
Our FIT-Print can also generalize to other types of datasets. Our primitive FIT-Print aims to optimize a perturbation $\bm{r}$ on the input $\bm{x}$ to make the mapping vector close to the targeted fingerprint. The main part of the loss function is as Eq.~(\ref{eq:main_part}).
\begin{equation}
\label{eq:main_part}
    \min\mathcal{L}(f(M_o(\bm{x}+\bm{r}), \bm{F}).
\end{equation}

However, for the discrete data ($e.g.$, text data), it is not feasible to directly add the perturbation to it. Thus, a rewriting function $g(\bm{x})$ can be introduced to rewrite the characters, words, or sentences. The loss function can be changed to Eq.~(\ref{eq:text}).

\begin{equation}
    \label{eq:text}
    \min\mathcal{L}(f(M_o(g(\bm{x})), \bm{F}).
\end{equation}
The main challenge lies in how to design an effective optimization method to find a rewriting function $g(\pmb{x})$ which minimizes the above loss function. There are already some existing works~\cite{guo2021gradient, wen2024hard} to fulfill this task. Accordingly, FIT-Print can be adapted to other data formats ($e.g.$, text or tabular).

% \begin{figure}[t]
%     \centering
%     \includegraphics[width=0.99\linewidth]{fig/visual-lambda.pdf}
%     \caption{The visualization of the original images and the perturbed testing samples with different $\lambda$.}
%     \label{fig:visual-lambda}
%     \vspace{-10pt}
% \end{figure}

\partitle{Case Study on Text Generation Model}
\label{sec:case}
We conduct a case study on implementing FIT-Print on text generation models. Text generation models~\cite{pang2025modelshield, guan2024sample} have become the most famous models in recent years and have been widely applied in various domains. Specifically, the text generation model predicts the next token in a sequence of tokens, \ie, the output of the text generation models is a sequence of logits. Given an input sequence $\bm{s}=\{s_1, s_2, ..., s_q\}$, where $q$ is the number of tokens in the sequence, and a vocabulary $\mathcal{V}$, the text generation model outputs a sequence $\bm{o}\in \mathbb{R}^{q\times |\mathcal{V}|}$. The $i$-th element in $\bm{o}$ is the probability logit of the tokens in the vocabulary.

To implement FIT-Print on text generation models, we need to optimize Eq.~(\ref{eq:text}) to generate the testing samples. Arguably, our FIT-ModelDiff and FIT-LIME can easily generalize to protect text generation models. Specifically, the mapping functions used in FIT-ModelDiff and FIT-LIME (\ie, cosine similarity and average entropy) can be directly applied to text generation models. This is because text generation models differ from classification models only in the output dimension, and these two functions are inherently able to calculate data with different dimensions. The main challenge is to optimize the discrete text data to minimize Eq.~(\ref{eq:text}). To achieve this goal, we can exploit existing text optimization methods~\cite{guo2021gradient, wen2024hard}. Specifically, we implement the optimization method proposed in \cite{wen2024hard}. It optimizes the embeddings of the text sequences and then finds the nearest token in the embedding space to replace the original token.

We further conduct experiments to verify the effectiveness of our FIT-ModelDiff and FIT-LIME on text generation models. We use two popular text generation models (\ie, GPT-2~\cite{radford2019language} and BERT~\cite{devlin2018bert}) and two datasets (\ie, ptb-text-only~\cite{marcus1993building} and lambada~\cite{paperno2016lambada}) for our case study. Table~\ref{tab:text} shows the bit error rates (BERs) of applying our methods to text generation models. The BERs are all lower than the threshold $\tau=0.227$, indicating that FIT-Print is also applicable to protect the IPR on other data formats.

%Take the text generation model as an instance. We implement the optimization method used in \cite{wen2024hard} to verify the effectiveness of FIT-Print. Specifically, we utilize two popular text generation models, GPT-2~\cite{radford2019language} and BERT~\cite{devlin2018bert}, and two datasets, ptb-text-only~\cite{marcus1993building} and lambada~\cite{paperno2016lambada}. Table~\ref{tab:text} shows the bit error rates (BER) of applying FIT-ModelDiff and FIT-LIME to text generation models and demonstrates that FIT-Print is also applicable to protect the IPR on other data formats.

\begin{table}[t]
    \tabcolsep=4mm
    \renewcommand{\arraystretch}{1.2}
    \centering
    \caption{The BERs of applying FIT-ModelDiff and FIT-LIME to text generation models (lower is better).}
    \label{tab:text}
    \scalebox{0.82}{
    \begin{tabular}{c|cc|cc}
    \hline
    \hline
    % \toprule
    Method$\rightarrow$ & \multicolumn{2}{c|}{FIT-ModelDiff} & \multicolumn{2}{c}{FIT-LIME} \\
    Dataset$\downarrow$ Model$\rightarrow$  &  GPT-2 & BERT & GPT-2 & BERT \\
    \hline
    ptb-text-only & 0.188 & 0.188 & 0.031 & 0.000 \\
    lambada & 0.188 & 0.125 & 0.062 & 0.000\\
         \hline
         \hline
         % \bottomrule
    \end{tabular}
    }
    % \vspace{-0.5em}
\end{table}

\subsection{The Overhead of FIT-ModelDiff and FIT-LIME}
\label{sec:overhead}

Compared with existing model fingerprinting methods, FIT-Print needs to optimize the testing samples and thus has an extra overhead. We hereby present a detailed analysis of the time and space complexity of the two fingerprinting methods, FIT-ModelDiff and FIT-LIME. FIT-ModelDiff and FIT-LIME are the representatives of bit-wise and list-wise methods, respectively. As such, they have different trade-offs in time and space overhead.

\partitle{Overhead during the Fingerprint Verification Stage}
\begin{itemize}
    \item For FIT-ModelDiff, the space complexity is $O(1)$ and the time complexity is $O(k)$ where k is the length of the targeted fingerprint. FIT-ModelDiff is a bit-wise fingerprinting method that extracts the fingerprint bit by bit. It only reads two samples to calculate one bit in the fingerprint in each step. The fingerprint can be extracted in k steps. Therefore, the space complexity is $O(1)$, and the time complexity is $O(k)$.
    \item For FIT-LIME, the space complexity is $O(k)$ and the time complexity is $O(k/\beta)$ where $\beta$ is the batch size. FIT-LIME is a list-wise method that extracts the fingerprint as a whole list. FIT-LIME needs to read all the samples at the same time to extract the fingerprint. On the other hand, FIT-LIME can calculate the outputs of an entire batch at the same time. As such, the space complexity is $O(k)$, and the time complexity is $O(k/\beta)$.
\end{itemize}

\begin{table}[t]
    \tabcolsep=1mm
    \renewcommand{\arraystretch}{1.2}
    \centering
    \caption{The time overhead (seconds) of each optimization iteration applying FIT-ModelDiff and FIT-LIME with different models. ``\#Params'' indicates the number of parameters.}
    \label{tab:time}
    \scalebox{0.83}{
    \begin{tabular}{c|ccccc}
    \hline
    \hline
    % \toprule[1.5pt]
    % Method & \multicolumn{3}{c|}{FIT-ModelDiff} & \multicolumn{3}{c}{FIT-LIME} \\
    Model  & MobileNetV2 & EfficientNet & ResNet18 & SwinViT & ViT\\
    \#Params (M) & 3.50 & 5.29 & 11.69 & 28.29 & 86.57 \\
    \hline
    FIT-ModelDiff & 4.30 & 5.06 & 4.07 & 5.25 & 6.57\\
    FIT-LIME & 0.84 & 1.24 & 0.74 & 1.34 & 2.13\\
         \hline
         \hline
         % \bottomrule[1.5pt]
    \end{tabular}
    }
    % \vspace{-0.5em}
\end{table}

\partitle{Overhead during the Testing Samples Extraction Stage}
In each iteration of optimizing the testing samples, we need to perform one forward propagation and one backward propagation of the fingerprint verification method. Assuming that we utilize $\xi$ augmented models during optimization, the time complexities of each iteration of testing sample extraction in FIT-ModelDiff and FIT-LIME are $O(\xi\cdot k)$ and $O(\xi\cdot k/\beta)$. $k$ is the length of the targeted fingerprint and $\beta$ is the batch size. We evaluate the time overhead during the testing samples extraction stage for different methods, as presented in Table~\ref{tab:time}. Overall, our proposed FIT-ModelDiff and FIT-LIME demonstrate high efficiency. Even with the large ViT model, the time cost is still low ($<7$s per iteration and $<0.6$h for the entire 300 iterations). Furthermore, the testing sample extraction process is a one-time operation for the protected model and does not require repeated execution. As such, our FIT-Print is efficient in optimization, and the overall overhead is acceptable.

%For instance, in our main experiments, we utilize $10$ augmented models and a $256$-bit targeted fingerprint. It takes only 3 seconds for one optimization iteration in FIT-ModelDiff and 1 second for that in FIT-LIME. As such, arguably, our FIT-ModelDiff and FIT-LIME are efficient in optimization and the overall overhead is acceptable.

\subsection{Potential Limitations and Future Directions}
\label{sec:limitations}

As the first attempt at the target fingerprinting paradigm to defeat false claim attacks, we must acknowledge that our method still has the following potential limitations.

First, our FIT-Print depends on a trustworthy third-party institution to register the fingerprint with a timestamp, which is not currently established. However, we argue that this will certainly be realized in the foreseeable future. For example, the existing intellectual property office (IPO) or artificial intelligence regulator (AIR) can be responsible for this duty~\cite{li2025rethinking}. First, it is common for developers to register their intellectual property, including valuable models, with the IPO for copyright protection. Second, many countries and regions are in the process of establishing or have established the AIR (\eg, as exemplified in the EU Artificial Intelligence Act) to ensure security and transparency before deploying the AI models. As such, it is also feasible for the AIR to manage model registrations and audit potential infringement.

Second, our FIT-Print fails to provide formal proof of the resistance to false claim attacks due to the complexity of DNNs. %As such, a more powerful adversary may still be able to conduct a successful false claim attack. 
We will investigate how to achieve a certified robust model fingerprinting method against false claim attacks with a rigorous conceptual guarantee in our future work.

Thirdly, as discussed in Section~\ref{sec:extension}, for deterministic models that do not involve randomness ($e.g.$, CNN and LLM), our FIT-Print can generalize to different types of models and datasets. However, for non-deterministic models that involve randomness ($e.g.$, diffusion models~\cite{ho2020denoising}), we have to admit that we do not know whether our methods and existing fingerprinting methods can be adapted to them since they have a completely different inference paradigm. We will further explore and discuss this intriguing problem in our future work.

\section{Conclusion}
\label{sec:conclusion}

% In this paper, we first revisited existing model fingerprinting methods, designed a false claim attack by crafting some transferably easy samples, and revealed that existing methods were vulnerable to the false claim attack. We found that the vulnerability can be attributed to the untargeted nature of existing methods, which compare the outputs of any given samples on different models rather than the similarities to specific signatures. To tackle the above issue, we proposed FIT-Print, a false-claim-resistant model fingerprinting paradigm. FIT-Print transformed the fingerprint of the model into a targeted signature by optimizing the testing samples. We correspondingly designed two fingerprinting methods based on FIT-Print, namely the bit-wise FIT-ModelDiff and the list-wise FIT-LIME. Our empirical experiments demonstrated the effectiveness, conferrability, and resistance to false claim attacks of our FIT-Print. We hope our FIT-Print can provide a new angle on model fingerprinting methods to facilitate more secure and trustworthy model sharing and trading.
In this paper, we revisited existing model fingerprinting methods and revealed their vulnerability to false claim attacks by designing an attack that crafts transferable, "easy" samples. We demonstrated that this vulnerability primarily stems from the untargeted nature of current approaches, which compare the outputs of arbitrary samples across models rather than evaluating their similarity to specific signatures. To address this issue, we proposed FIT-Print, a targeted, false-claim-resistant model fingerprinting paradigm. By optimizing testing samples, FIT-Print effectively transforms the model fingerprint into a specific, verifiable signature. Based on this framework, we developed two corresponding methods: the bit-wise FIT-ModelDiff and the list-wise FIT-LIME. Extensive empirical experiments verified the effectiveness, conferrability, and robust resistance of FIT-Print against false claim attacks. Ultimately, we hope that FIT-Print offers a new perspective on model fingerprinting, facilitating more secure and trustworthy model sharing and trading.

\section*{Acknowledgements}

This research is supported by the Kunpeng-Ascend Science and Education Innovation Excellence/Incubation Center, the National Natural Science Foundation of China under Grants (62441238, U2441240), and the Zhejiang Provincial Natural Science Foundation of China (No. LD24F020010). The authors sincerely thank Prof. Kui Ren from Zhejiang University for the helpful comments and suggestions on an early draft of this paper, and Dr. Najeeb Jebreel from Universitat Rovira i Virgili for his assistance with language editing and proofreading.

% This research is supported in part by the National Key Research and Development Program of China under Grant 2021YFB3100300 and the National Natural Science Foundation of China under Grants (62441238, 62072395, and U20A20178). The authors sincerely thank Prof. Kui Ren from Zhejiang University for the helpful comments and suggestions on an early draft of this paper.

\bibliographystyle{IEEEtran}
\bibliography{ref}

\appendix

\setcounter{proposition}{0}
\setcounter{equation}{0}
\setcounter{page}{1}

\subsection{The Proof of Proposition 1}
\label{sec:proof}

\begin{proposition}
    % \label{theorem:threshold}
    Given the security parameter $\kappa$ and the fingerprint $\bm{F}\in\{-1, 1\}^k$, if $\tau$ satisfies that
    \begin{equation}
        \sum_{d=0}^{\lfloor \tau k \rfloor}\tbinom{k}{d} (\frac{1}{2})^k\leq \kappa,
    \end{equation}
    where $\tbinom{k}{d}=k!/[d!(k-d)!]$, the probability of a successful false claim attack, \ie, the BER is less than $\tau$ with the adversaries testing samples, is less than $\kappa$.
\end{proposition}

\begin{proof}
    As mentioned in Section II-B and Section II-C, registering the fingerprint with a timestamp can avoid any subsequent false claim attack. As such, the adversary needs to craft the testing samples $\bar{\mathcal{X}}_T$ which can be transferred to independent models in advance. We assume that the adversary extracts the fingerprint $\bar{\bm{F}}$ from the independent model $M_I$ using the testing samples $\bar{\mathcal{X}}_T$, as follows.
    \begin{equation}
        \bar{\bm{F}}={\tt sign}(M_I(\bar{\mathcal{X}}_T)),
    \end{equation}
    We assume that $\breve{\bm{F}}$ denotes the adversary's target fingerprint. Since $\bar{\bm{F}} \in \{-1, 1\}^k$ is a $k$-bit binary vector and the adversary has no knowledge of the independent model $M_I$, the probability of any bit in $\bar{\bm{F}}$ to match the corresponding bit in $\breve{\bm{F}}$ is $1/2$. Thus, to satisfy Eq.~(5) in Definition~4, \ie, making the BER between $\bar{\bm{F}}$ and $\breve{\bm{F}}$ less than $\tau$, there needs to have at least $k - \lfloor \tau \cdot k \rfloor$ bits in $\bar{\bm{F}}$ match $\breve{\bm{F}}$. Based on the binomial theorem, we have the probability of the aforementioned scenario, \ie, a successful false claim attack, is as follows.
    \begin{equation}
    \label{eq:prob}
        \mathrm{Pr}\{{\tt BER}(\bar{\bm{F}}, \breve{\bm{F}}) \leq \tau\}=\sum_{d=0}^{\lfloor \tau k \rfloor}\tbinom{k}{d} (\frac{1}{2})^k.
    \end{equation}
    Since the right-hand side of Eq.~(3) is less than $\kappa$, the probability of a successful false claim attack, \ie, the BER, is less than $\tau$ with the adversarial testing samples, is also less than $\kappa$.
\end{proof}

% \section{Formulation of the Design Objectives}

% \section{Pseudocode of FIT-Print}

% 这里放testing sample extraction, ownership verification, FIT-ModelDiff和FIT-LIME的伪代码，一共四个

% \section{Analysis on How to Select the Threshold $\tau$}

% \subsection{Implementation Details}
% \label{sec:impdetails}

\subsection{Details of the Model Reuse Techniques}
\label{sec:reusedetails}

In our experiments, we evaluate FIT-Print and other model fingerprinting methods against the following five categories of model reuse techniques, including copying, fine-tuning, pruning, model extraction, and transfer learning.

\begin{itemize}
    \item \textbf{Copying}: Copying refers to the adversary somehow gaining white-box access to the parameters and architecture of the victim model. Subsequently, the adversary steals the model by directly copying it. It may occur when the model is open-source and publicly available.
    \item \textbf{Fine-tuning}: Fine-tuning means the adversary re-trains the victim model with its own dataset, which is related to the primitive task of the victim model. We consider three types of fine-tuning, 
    %denoted as {\tt Fine-tuning(10\%)}, {\tt Fine-tuning(50\%)}, and {\tt Fine-tuning(100\%)}, which means 
    \ie, we fine-tune the last 10\%, 50\%, and 100\% layers of the victim models.
    \item \textbf{Pruning}: Pruning~\cite{cheng2024survey} intends to compress the model by removing redundant parameters. We leverage weight pruning as our pruning method, which prunes the neurons according to their activations. We prune 10\%, 30\%, and 50\% of the parameters of the victim model. %denoted as {\tt Pruning(10\%)}, {\tt Pruning(30\%)}, and {\tt Pruning(50\%)} respectively.
    \item \textbf{Model Extraction}: Model extraction attempts to steal the knowledge of the victim model via only black-box access. The adversary can obtain the output of the victim model to train the extracted model. Following \cite{jia2022ZestLIME} and \cite{li2021ModelDiff}, we conduct the model extraction by optimizing the following loss function~\cite{yang2022rethinking}:
    \begin{equation}
        \min_{M_e} \mathcal{L}_{CE}(M_e(x), M_o(x)),
    \end{equation}
    where $\mathcal{L}_{CE}$ is the cross-entropy loss, $M_e$ and $M_o$ are the extraction source model and the original model.
    In our experiments, we implement the model extraction in two different scenarios. {\tt Extract(same)} and {\tt Extract(different)} refer to utilizing the same or different model architectures to extract the source model, respectively.
    \item \textbf{Transfer Learning}: Transfer learning~\cite{zhuang2020comprehensive} is an ML technique where the victim model trained on one task is adapted as the starting point for a model on a second related task. In our experiments, we replace the last layer of the model to fit the second task and fine-tune the model for 200 epochs. Similar to the setting of fine-tuning, we fine-tune the last 10\%, 50\%, and 100\% layers of the victim models to implement transfer learning. These models are denoted as {\tt Transfer(10\%)}, {\tt Transfer(50\%)}, and {\tt Transfer(100\%)} respectively.
\end{itemize}

\end{document}